\newif\ifec
\ecfalse

\ifec

\documentclass[format=acmsmall, review=false]{acmart}
\usepackage{ec24/acm-ec-24-proc}
\usepackage{booktabs} 
\usepackage[ruled]{algorithm2e} 
\usepackage{natbib} 

\SetAlFnt{\small}
\SetAlCapFnt{\small}
\SetAlCapNameFnt{\small}
\SetAlCapHSkip{0pt}
\IncMargin{-\parindent}

\setcitestyle{authoryear}

\ccsdesc[500]{Theory of computation~Models of learning}

\keywords{Causal Inference, Experimental Design, Sample Size, Statistical Decision Theory, Treatment Heterogeneity}

\acmYear{2024}\copyrightyear{2024}
\setcopyright{acmlicensed}
\acmConference[EC '24]{Conference on Economics and Computation}{July 8--11, 2024}{New Haven, CT, USA}
\acmBooktitle{Conference on Economics and Computation (EC '24), July 8--11, 2024, New Haven, CT, USA}
\acmDOI{10.1145/3670865.3673458}
\acmISBN{979-8-4007-0704-9/24/07}

\title{Minimax-Regret Sample Selection in Randomized Experiments}

\author{Yuchen Hu}
\affiliation{%
  \institution{Stanford University}
  \city{Stanford}
  \state{CA}
  \country{USA}
}
\author{Henry Zhu}
\affiliation{%
  \institution{Stanford University}
  \city{Stanford}
  \state{CA}
  \country{USA}
}
\author{Emma Brunskill}
\affiliation{%
  \institution{Stanford University}
  \city{Stanford}
  \state{CA}
  \country{USA}
}
\author{Stefan Wager}
\affiliation{%
  \institution{Stanford University}
  \city{Stanford}
  \state{CA}
  \country{USA}
}

\begin{abstract}
Randomized controlled trials are often run in settings with many subpopulations that may
have differential benefits from the treatment being evaluated. We consider the problem of sample
selection, i.e., whom to enroll in a randomized trial, such as to optimize welfare in a heterogeneous population.
We formalize this problem within the minimax-regret framework, and derive optimal sample-selection
schemes under a variety of conditions.
Using data from a COVID-19 vaccine trial, we also highlight how different objectives and decision rules
can lead to meaningfully different guidance regarding optimal sample allocation.
\end{abstract}

\else

\documentclass{article}
\usepackage{lineno}

\usepackage{amsmath, amsthm, amssymb}
\usepackage{graphicx}
\usepackage{verbatim}
\usepackage{natbib}
\usepackage{caption}
\usepackage{subcaption}
\usepackage{fancyvrb}
\usepackage{enumerate}
\usepackage{relsize}
\usepackage{algorithm}
\usepackage{algpseudocode}
\usepackage{enumitem}
\usepackage{mathtools}
\usepackage{rotating}

\usepackage{hyperref}
\usepackage[margin=1.5in]{geometry}
\hypersetup{colorlinks,citecolor=blue,urlcolor=blue,linkcolor=blue}
\graphicspath{{./figures/}}

\theoremstyle{plain}
\newtheorem{proposition}{Proposition}[section]

\newtheorem{lemma}[proposition]{Lemma}
\newtheorem{theorem}[proposition]{Theorem}

\theoremstyle{definition}

\newtheorem{definition}{Definition}

\theoremstyle{remark}

\title{Minimax-Regret Sample Selection\\ in Randomized Experiments}
\author{Yuchen Hu \and Henry Zhu \and Emma Brunskill \and Stefan Wager}
\date{Stanford University}

\fi

\usepackage{multirow}
\usepackage{notations}
\usepackage{tikz}
\usetikzlibrary{shapes.geometric, arrows}

\tikzstyle{procedure} = [rectangle, rounded corners, 
minimum width=2cm, 
minimum height=0.8cm,
text width=2cm,
text centered, 
draw=black]

\tikzstyle{procedure_larger} = [rectangle, rounded corners, 
minimum width=2.5cm, 
minimum height=0.8cm,
text width=2.5cm,
text centered, 
draw=black]

\tikzstyle{procedure_grey} = [rectangle, rounded corners, 
minimum width=2cm, 
minimum height=0.8cm,
text width=2cm,
text centered, 
draw=black, 
fill=gray!20]

\tikzstyle{procedure_larger_grey} = [rectangle, rounded corners, 
minimum width=3cm, 
minimum height=0.8cm,
text width=3cm,
text centered, 
draw=black, 
fill=gray!20]

\tikzstyle{procedure_ellipse} = [ellipse, rounded corners, 
minimum width=1.5cm, 
minimum height=0.8cm,
text width=1.5cm,
text centered, 
line width=1.5pt,
draw=black]

\tikzstyle{arrow} = [thick,->,>=stealth]

\tikzstyle{arrow_darkred} = [->, >=stealth, line width=2pt, color=darkred]

\definecolor{darkred}{RGB}{139,0,0} 

\begin{document}

\ifec

\maketitle

\else

\maketitle

\begin{abstract}
Randomized controlled trials are often run in settings with many subpopulations that may
have differential benefits from the treatment being evaluated. We consider the problem of sample
selection, i.e., whom to enroll in a randomized trial, such as to optimize welfare in a heterogeneous population.
We formalize this problem within the minimax-regret framework, and derive optimal sample-selection
schemes under a variety of conditions.
Using data from a COVID-19 vaccine trial, we also highlight how different objectives and decision rules
can lead to meaningfully different guidance regarding optimal sample allocation.
\end{abstract}

\fi

\section{Introduction}
\label{sec:intro}

Results from randomized controlled trials (RCTs) are often used to guide decision rules,
especially in medical settings. It is thus natural to ask: How should a researcher
design an RCT if their primary objective is to optimize for the expected welfare achieved
via the induced decision rule? How should they reason about whom to enroll in the study?

The goal of this paper is to flesh out and investigate properties and tradeoffs for
welfare-optimizing RCT designs in heterogeneous populations, i.e., in populations
where different decision rules (also often referred to as treatment strategies) may be appropriate for different subgroups.
There is a large existing  literature that gives guidelines on {\it how many} study participants to
enroll in a RCT. Such sample size calculations are usually driven by power considerations,
and aim to guarantee statistically significant detection of the average treatment effect
whenever this effect is larger than a practically relevant  threshold
\citep{altman1980statistics,moher1994statistical}.
More recently, \citet{manski2016sufficient,manski2019trial} revisited sample size calculations
with a focus on welfare considerations, and showed that relatively small sample sizes can still
yield strong welfare guarantees under a minimax-regret criterion.
\citet{azevedo2020b,azevedo2023b} discuss properties of Bayesian rules for choosing the sample
size. Perhaps surprisingly, however, the question of {\it whom} to enroll, or how to allocate a
given number of spots in a study across different eligible groups, has not received similar attention.

In this paper, we consider sample selection for RCTs under the following simple model.
The population of eligible study participants can be divided into $g = 1, \, \ldots, \, G$
groups, each with its own treatment effect $\tau_g$. The researcher has budget to run an
RCT with $N$ total study participants, and needs to choose how many people $n_g$ to enroll from each
group $g$.  The researcher will use the data from the RCT to propose group-specific treatment rules. 
The researcher's ultimate goal is to achieve robust welfare guarantees and they seek to choose
$n_g$ in line with this goal.  
For example, consider a company who wishes to make product design decisions and is not interested in precise estimates of treatment effects but instead in using a limited experimental budget to allocate samples to customer subgroups in order to  optimize for downstream outcomes under an induced decision policy. Formally, as in \citet{manski2016sufficient,manski2019trial} we specify 
the objective through the lens of minimax regret \citep{savage1951theory,manski2004statistical}.
We provide further details on our model in Section \ref{sec:minimax_selection}.

Perhaps our most surprising finding is as follows. First let $U_g$ denote the $g$-th group's expected utility.
Suppose we can assign different treatments to different groups,
and seek to optimize the weighted average utility \smash{$U = \sum_{g = 1}^G \alpha_g U_g$} where the
weights $\alpha_g$ are proportional to the size of each group in the population.\footnote{This aggregated
utility objective is a scaled version of utilitarian social welfare that weights everyone in the population
equally \citep{moulin2004fair}.} Then, the minimax-regret study design allocates samples to the $g$-th group proportionally
to \smash{$\alpha_g^{2/3}N$}. In particular, this recommendation deviates from the conventional practice of selecting samples
in a stratified experiment in proportion to each group's share in the population
\citep{singh2013fundamentals,sharma2017pros}, and instead over-samples minority groups.
In Section \ref{sec:other_selections}, we discuss reasons for this finding. We also prove
that more familiar sample-allocation schemes, like proportional sampling, are   
minimax optimal under alternate decision theoretic framings. We also consider a hypothetical
application of our method guided by historical COVID-19 trial data, which illustrates the notably different sample allocations that arise under different framings (Section~\ref{sec:case_study}). Overall, our results
highlight how different welfare-focused objectives have different implications as to optimal
sample selection.

\subsection{Related work}

There has been considerable work on methods for learning treatment rules in
heterogeneous populations, using data from either randomized or observational studies
\citep{athey2021policy,kitagawa2018should,manski2004statistical,stoye2009minimax,swaminathan2015batch,zhao2012estimating}.
These methods assume that the researcher has access to a dataset of a given size,
and then seeks to use the available data to discover heterogeneity patterns and learn
treatment assignment rules. For example, \citet{ross2023estimated} use electronic-health-record
data from the US Department of Veterans Affairs to study heterogeneity in how different groups
of patients with suicide ideation or recent suicide attempts respond to hospitalization,
and argue that decision rules that leverage such heterogeneity could substantially improve outcomes.
Here, in contrast, we do not take the dataset as given, and  instead seek
to provide guidance on data collection with heterogeneous study populations.

Our research question falls within literature on using decision theory
to guide the design of RCTs. As noted above,
\citet{manski2016sufficient,manski2019trial} and \citet{azevedo2020b,azevedo2023b}
use decision theoretic criteria to choose the sample size in an RCT.
Another line of work compares different randomization rules over the treatment and control sample allocation in terms of the resulting
variance for the average treatment effect, and seeks intervention randomization rules with
minimax variance \citep{bai2021randomize,kallus2021optimality,li1983minimaxity,wu1981robustness}.
\citet{banerjee2020theory} study the value of randomization itself under a number
of different decision theoretic frameworks. However, we are not aware of prior
work in this literature that considers the question of sample allocation across groups.

Finally, our work also relates to the broader literature on fairness in data analysis.
RCTs have often been run in blatantly unfair ways. For example, until recently, a substantial fraction of medical
research in the United States was conducted on white men, while excluding women and racial
minorities \citep{dresser1992wanted}; and FDA-approved trials still under-sample black
participants relative to their share of the population \citep{alsan2022representation}.
However, while it is easy to agree that some ways of running randomized trials are unfair, it can be
much harder to say what makes a trial fair. For example, if there is a minority subgroup that has
the potential for both disproportionate benefits and disproportionate risks from participating
in a study, what considerations should go into deciding how to fairly include them in the study?
In such settings results on welfare-driven sample selection may help provide useful insights
into potential desiderata for fair RCT design. Similar considerations are also relevant to
fairness-aware data collection when training machine learning and AI systems \citep{holstein2019improving}.

\section{Minimax-Regret Sample Selection}
\label{sec:minimax_selection}

We start with a standard analysis in a stratified completely randomized controlled trial where participants within a stratum receive treatment or control in a 1:1 ratio. The experimenter has the freedom to enroll a maximum of $N$ participants from $G$ groups, with $n_g$ participants allocated to group $g$, $g=1,\dots, G$. In other words, the class of feasible sample selection $\mathcal{N}$ is composed of all $\bfn=(n_1,\dots,n_G)^\top$ such that $\sum_{g=1}^G n_g\le N$, $n_g\in 2\ZZ_{\ge 0}, \forall g$.

Given the sample selection $\bfn$, the experimenter conducts the experiment, and collects data $D=\cb{Y_g,W_g}_{g=1,\dots,G}$ on all of the groups from the experiment, in which $D$ is generated from some data-generating process $\mathcal{D}(\bfn,\tau)$ parameterized by $\tau$, where $Y_g=(Y_{g,1},\dots,Y_{g,n_g})^\top$ is the vector of observed outcomes and $W_g=(W_{g,1},\dots,W_{g,n_g})^\top$ is the vector of treatment assignments. Following the classic potential outcome framework,  under the SUTVA condition \citep{imbens2015causal}, the observed outcome $Y_{g,i}$ can be written as $Y_{g,i}=W_{g,i}Y_{g,i}(1)+(1-W_{g,i})Y_{g,i}(0)$, where $W_{g,i}$ is the treatment assigned to unit $i$ in group $g$,  and $Y_{g,i}(w)$, $w\in\cb{0,1}$, is the potential outcome that we would have observed from unit $i$  in group $g$ if it was assigned to treatment group $w$. 

Based on the collected dataset, the experimenter proceeds to estimate the treatment effect in each group with an estimator $\htau(D)$ that maps from the collected data $D$ to $G$ real-valued estimates of treatment effect $\htau_g$.
They then make a decision $\delta(D)=(\delta_1,\dots,\delta_G)^\top\in \cb{0,1}^G$ regarding the administration of treatment within each group.
Based on the quality of the decisions, 
the experimenter attains an aggregated total utility $U(\alpha,\tau,\delta(D))$ defined as
\begin{equation}
U(\alpha,\tau,\delta(D)) = \sum_{g=1}^G \alpha_g\cdot\tau_g\delta_g.
\label{eq:utility}
\end{equation}
The corresponding regret of making such a decision is
\begin{equation}
R(\bfn,\delta(D))= U(\alpha,\tau,\delta^*) - U(\alpha,\tau,\delta(D)).
\end{equation}
where
\begin{equation}
\delta^* \in \argmax_{\delta\in \cb{0,1}^G}U(\alpha,\tau,\delta).
\end{equation}

\begin{figure}[t!]
\centering
\begin{tikzpicture}[node distance=2cm]
\node (ss) [procedure] {Sample Selection\\ $\bfn\in\mathcal{N}$};
\node (dgp) [procedure, below of=ss] {Data Generating Process\\ $\tau\in \RR^G$};
\node (estimator) [procedure, above of=ss] {Estimator\\ $\htau(D)$};
\node (nature) [procedure_ellipse, left of=ss, xshift=-2cm] {Adversary};
\node (data) [procedure_larger, right of=ss, xshift=1.5cm] {Data\\ $D\sim \mathcal{D}(\bfn,\tau)$};
\node (regret) [procedure, right of=data, xshift=1.5cm] {Regret $R(\bfn,\delta(D))$};
\node (estimation) [procedure_larger, above of=data] {Estimated ATE\\ $\p{\htau_1,\dots,\htau_G}^\top$};
\node (policy) [procedure_larger, above of=regret] {Treatment Policy\\ $\delta(D)$};
\node (eregret) [procedure_larger_grey, below of=regret] {Expected Regret $\EE{R(\bfn,\delta(D))}$};

\draw [arrow, dashed] (nature) -- (ss)node[midway, above, yshift=6.5pt] {know};
\draw [arrow, dashed] (nature) -- (estimator);
\draw [arrow] (nature) -- (dgp)node[midway, below, yshift=-3pt] {choose};
\draw [arrow] (dgp) -- (data);
\draw [arrow] (estimator) -- (estimation);
\draw [arrow] (ss) -- (data);
\draw [arrow] (data) -- (estimation);
\draw [arrow] (estimation) -- (policy);
\draw [arrow] (policy) -- (regret);
\draw [arrow] (dgp) -- (regret);
\draw [arrow, dashed] (regret) -- (eregret);
\end{tikzpicture}
\caption{A flow chart illustrating the choice of a sample selection within a minimax framework where the data generating process is chosen by an adversary.}
\label{fig:flow_minimax}
\end{figure}
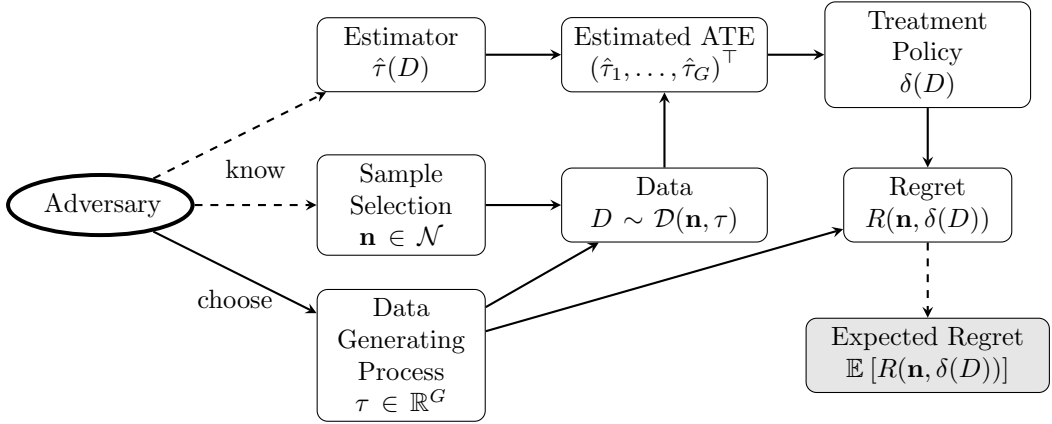

Ideally, one would like to identify the optimal sample selection $\bfn\in\mathcal{N}$ by minimizing the expected total regret as given by $\EE[D]{R(\bfn,\delta(D))}$, with expectation taken over the data generation. However, the expected total regret is contingent upon the unknown quantity $\tau$, rendering it challenging to quantify during the design phase before observing any data. 

Here we adopt a minimax-regret design approach where the sample selection $\bfn$ is chosen to minimize the maximum expected regret that could be incurred if $\tau$ is chosen by an adversary with knowledge of both the choice of sample selection $\bfn$ and the decision rule $\delta(D)$. An illustration of this strategy can be found in Figure \ref{fig:flow_minimax}. Under this framework, we say that a sample selection rule $\bfn\in\mathcal{N}$ is minimax-regret if this sample selection, together with the choice of estimator, minimizes $\EE[D]{R(\bfn,\delta(D))}$ when the data generation is adversarial.
For simplicity, we only consider sample allocation rules where each group is given
an even number of samples; this lets us focus on completely randomized experiments with
treatment (and control) perfectly balanced in each stratum.

\begin{definition}
Let \smash{$\mathcal{N}=\{\bfn:\sum_{g=1}^G n_g\le N, n_g\in 2\ZZ_{\ge 0}, \forall g \}$}.
The worst-case regret of any sample selection rule $\bfn \in\mathcal{N}$ is
\begin{equation}
\label{eq:minimax_def}
H(\bfn) = \inf_{\delta(D)}\max_{\tau\in\RR^{G}} \EE[D]{R(\bfn,\delta(D))}.
\end{equation}
Such a sample selection is minimax-regret if $\bfn \in \argmin_{\bfn\in\mathcal{N} }H(\bfn)$.
\label{definition:minimax}
\end{definition}

Below, we give a near-optimal solution to this minimax-regret problem when the data-generating distribution $D$ is from a Gaussian\footnote{When the data-generating distribution is generalized to be from a sub-Gaussian class, an asymptotic limits-of-experiments analogue to Theorem \ref{theorem:minimax} holds in the large sample limit; see \citet{hirano2009asymptotics} for a general discussion of results of this type.}
class. The sub-optimality in our result is only due to rounding to achieve integer sample allocations, and our solution is optimal whenever no rounding is needed in \eqref{eq:minimax_selection}.
We present a brief proof of the theorem, with proof of the technical lemmas deferred to the supplementary material. 
Compared with a conventional proportional selection, the minimax selection rule dilutes the influence of the weights $\alpha_g$, and oversamples the minority groups.

\begin{theorem}
Suppose $Y_{g,i}(0)$ and $Y_{g,i}(1)$ are generated i.i.d. from Gaussian distributions $N(b_g-\tau_g/2,s_{0,g}^2)$ and $N(b_g+\tau_g/2,s_{1,g}^2)$, respectively.\footnote{For simplicity, we assume the experimenter knows the standard errors. However, $s_{0,g}$ and $s_{1,g}$ can also be regarded as upper bounds on unknown standard errors, as the adversary would maximize the noise levels.}
Suppose that $s_{0,g}^2,s_{1,g}^2,\alpha_g$ are bounded away from $0$ and $\infty$ for all $g$. 
Then among 1:1 completely stratified designs with at most $N$ total units,
the sample selection $\bfn^*=\p{n^*_1,n^*_2,\cdots,n^*_G}^\top$ with
\begin{equation}
n^*_g=2\left\lfloor\frac{(s_{0,g}^2+s_{1,g}^2)^{1/3}\alpha_g^{2/3}N}{2\sum_{g'=1}^G(s_{0,g'}^2+s_{1,g'}^2)^{1/3}\alpha_{g'}^{2/3}}\right\rfloor, \qquad\qquad g=1,\dots,G
\label{eq:minimax_selection}
\end{equation}
is nearly minimax-regret in the sense that
\begin{equation}
\frac{H\p{\bfn^*}-\min_{\bfn\in\mathcal{N} }H\p{\bfn}}{\min_{\bfn\in\mathcal{N} }H\p{\bfn}} = \oo\p{N^{-1}}.
\label{eq:minimax_gap}
\end{equation}
Furthermore, the decision rule induced by a threshold of the difference-in-means estimator
\begin{equation}
\begin{split}
&\delta^{DM}(D) = \p{I(\htau^\text{DM}_1> 0),\cdots,I(\htau^\text{DM}_G> 0)}^\top\\
&\htau^\text{DM}_g = \frac{2}{n_g}\sum_{\cb{i:W_{g,i}=1}} Y_{g,i} - \frac{2}{n_g}\sum_{\cb{i:W_{g,i}=0}} Y_{g,i},\qquad\qquad g=1,\dots,G
\end{split}
\label{eq:dim}
\end{equation}
is the minimax-regret decision rule given data from the experiment.
\label{theorem:minimax}
\end{theorem}

\begin{proof}[Proof of Theorem \ref{theorem:minimax}]
We start by verifying the last part of the statement, i.e., that given the available
data we should choose a decision based on the decision rule induced by a threshold of the difference-in-means estimator, and the induced regret is linear in standard errors.
This result is common in statistical treatment choice with Gaussian errors; see \citet{stoye2012minimax} and \citet{tetenov2012statistical}
for a discussion.

\begin{lemma}
Under the assumptions in Theorem \ref{theorem:minimax} and for any $\bfn \in \mathcal{N}$ with $n_g>0$ for all $g$, there exists a universal constant $C_0\approx 0.17$ such that
\begin{equation}
\label{eq:regbound}
H(\bfn)=C_0\sum_{g=1}^G \alpha_g\sqrt{\frac{2(s_{0,g}^2+s_{1,g}^2)}{n_g}}.
\end{equation}
Furthermore, the thresholded difference-in-means
estimator attains this bound,
\begin{equation}
\begin{split}
&\EE[D]{R(\bfn,\delta^\text{DM}(D))} \leq C_0\sum_{g=1}^G \alpha_g\sqrt{\frac{2(s_{0,g}^2+s_{1,g}^2)}{n_g}}
\end{split}
\end{equation}
for all $\tau \in \RR^G$.
\label{lemma:dm}
\end{lemma}

Now, to lower-bound minimax regret for sample allocation rules $\bfn \in \mathcal{N}$,
we start by considering the minimum of the regret bound \eqref{eq:regbound} over the
the following convex relaxation of $\mathcal{N}$
$$\widetilde{\mathcal{N}}=\cb{\bfn:\sum_{g=1}^G n_g\le N, n_g\in \RR_{\ge 0}, \forall g }. $$
Doing so yields the following:

\begin{lemma}
Under the assumptions in Theorem \ref{theorem:minimax},
\begin{equation}
\begin{split}
&\Tilde{n}_g = \argmin_{\bfn\in\widetilde{\mathcal{N}}}
C_0\sum_{g=1}^G \alpha_g\sqrt{\frac{2(s_{0,g}^2+s_{1,g}^2)}{n_g}} \\
&\Tilde{n}_g=\frac{(s_{0,g}^2+s_{1,g}^2)^{1/3}\alpha_g^{\frac{2}{3}}N}{\sum_{g'=1}^G(s_{0,g'}^2+s_{1,g'}^2)^{1/3}\alpha_{g'}^{2/3}},\qquad\qquad g=1,\dots,G.
\end{split}
\end{equation}
\label{lemma:selection}
\end{lemma}

Our proposed sample allocation rule $n^*_g$ is obtained by rounding $\Tilde{n}_g$.
\begin{lemma}
Under the assumptions in Theorem \ref{theorem:minimax},
\begin{equation}
\begin{split}
\inf_{\delta(D)}\max_{\tau\in\RR^{G}}\EE[D]{R(\bfn^*,\delta(D))}-\min_{\bfn\in\mathcal{N} }\inf_{\delta(D)}\max_{\tau\in\RR^{G}} \EE[D]{R(\bfn,\delta(D))}
= \oo\p{N^{-3/2}}.
\label{eq:proof_gap_rate}
\end{split}
\end{equation}
\label{lemma:gap_rate}
\end{lemma}

From Lemmas \ref{lemma:dm} and \ref{lemma:selection}, $\min_{\bfn\in\mathcal{N} }H\p{\bfn}=\Theta\p{N^{-1/2}}$. 
Combining this with \eqref{eq:proof_gap_rate} gives rise to the claimed result.
\end{proof}


Intuitively, the reason that the optimal sampling scheme in this setting samples the smaller proportion group more, is that,  in the absence of additional structure,  adding a small number of additional samples to  a small number of samples, is more helpful in reducing uncertainty compared to adding those same samples to a larger pool of data for the group with a higher proportion. 

\section{Conventional Sample Selections as Minimax-Regret Solutions}
\label{sec:other_selections}

Above, we found that the minimax-regret sample selection suggests oversampling minority groups and employing a sample allocation proportional to
\smash{$\alpha_g^{2/3}$}, which is not a common choice in the literature---at least currently.
In this section, we discuss two widely adopted sample selection methods, proportional selection and equal selection. We demonstrate that these common selections can also be regarded as regret-minimax selections when the experimenter makes a single  treatment decision for all groups,  or, respectively, uses an alternate  utility function from the one that we used in Section \ref{sec:minimax_selection}.

\subsection{Proportional Selection}
\label{subsec:prop_selection}

Proportional sampling (or sometimes referred to as ``stratified sampling''), is a longstanding and widely embraced approach in the literature of survey sampling \citep{singh2013fundamentals,sharma2017pros}, and is often the default choice when designing stratified experiments.
It involves sampling in each group proportionally to the group’s size when compared to the population. 
While the minimax solution to our defined problem in Section \ref{sec:minimax_selection} is different from the conventional proportional selection, we demonstrate below that proportional sampling is in fact minimax-regret in a different model where the researcher is constrained to make a single decision for the entire population based on the sign of a difference-in-means estimator calculated over the complete sample.

\begin{definition}
Consider the decision class $\{\Tilde{\delta}(D)\}$ that maps from the collected data $D$ to a single decision in $\cb{0,1}$,
i.e., only a joint decision can be made for all groups. Define $\Tilde{\delta}^{DM}(D)$ to be such a decision that is made according to the sign of a pooled difference-in-means estimator, i.e., 
\begin{equation}
\begin{split}
&\Tilde{\delta}^{DM}(D) = I(\htau^\text{DM,p}> 0),\\
&\htau^\text{DM,p} = \frac{2}{N}\sum_{g=1}^G\sum_{\cb{i:W_{g,i}=1}} Y_{g,i} - \frac{2}{N}\sum_{g=1}^G\sum_{\cb{i:W_{g,i}=0}} Y_{g,i}.
\end{split}
\end{equation}
The worst-case regret of any sample selection rule $\bfn \in\mathcal{N}$ subject to this class of decision rules is
\begin{equation}
\label{eq:joint_minimax}
H^\dag(\bfn) = \max_{\tau\in\RR^{G}} \EE[D]{R(\bfn,\Tilde{\delta}^{DM}(D))}.
\end{equation}
Such a sample selection is joint-minimax-regret if $\bfn \in \argmin_{\bfn\in\mathcal{N} }H^\dag(\bfn)$.
\label{definition:joint_minimax}
\end{definition}

\begin{theorem}
Under the conditions of Theorem \ref{theorem:minimax}, suppose in addition that $\alpha_gN$ is an even integer for all $g$.
Then the sample selection $\bfn^{\dag}=\p{n^\dag_1,n^{\dag}_2,\cdots,n^{\dag}_G}^\top$ with
\begin{equation}
{n}^{\dag}_g=\alpha_gN, \qquad\qquad g=1,\dots,G
\label{eq:proportional_selection}
\end{equation}
is joint-minimax-regret in the sense that
\begin{equation}
{n}^{\dag} \in \argmin_{\bfn\in\mathcal{N} }H^\dag(\bfn).
\end{equation}
\label{theorem:minimax_proportional}
\end{theorem}

The intuition behind this sample selection is that, when the decision-maker has to make a joint decision, any mismatch between the weight $\alpha_g$ in the population and the weight $n_g/N$ used in the estimator can lead to arbitrarily bad performance. Thus, under the decision process outlined in Definition \ref{definition:joint_minimax}, any sample selection other than proportional selection may result in infinite regret.

\subsection{An Egalitarian Solution}
\label{subsec:egalitarian_selection}

Although much work in learning and data-driven decision making is implicitly underpinned by a utilitarian ethical model,
some applications may call for different ways of comparing benefits across people or groups. In particular, some researchers
have argued that maximization of utilitarian social welfare may result in decisions that prioritize the overall population's
well-being at the expense of utility for minority groups \citep{rawls1971atheory}. 
In consideration of this, we also investigate minimax-regret sample selection guided by the egalitarian rule (also known as the max-min rule or the Rawlsian rule), which seeks to maximize the welfare of the worst-off individual in society \citep{rawls1971atheory,kolm1997justice,moulin2004fair,sen2018collective}. Under this setup, the regret objective is no longer
framed in terms of the average utility across the entire population, but rather at the group level---and we then focus on minimizing
regret uniformly across all groups.

\begin{definition}
Define $U_g(\tau_g,\delta_g) = \tau_g\delta_g$, and $R_g(\bfn,\delta(D))=U_g(\tau_g,\delta_g^*)-U_g(\tau_g,\delta_g(D))$.
The worst-case regret of any sample selection rule $\bfn \in\mathcal{N}$ for the egalitarian social welfare\footnote{Here, we focus on the single group with the smallest forecasted regret while averaging over the randomness in the data generation. Alternatively, one may consider defining $g^* = \argmax_g R_g(\bfn,\htau(D))$, where the worst-off group might vary with different realizations. While this is a valid formulation to consider, we do not pursue this path here.}  is
\begin{equation}
\label{eq:ega_minimax}
H^\ddag(\bfn) = \inf_{\delta(D)}\max_{\tau\in\RR^{G}}\max_g \EE[D]{R_g(\bfn,\delta(D))}.
\end{equation}
Such a sample selection is egalitarian-minimax-regret if $\bfn \in \argmin_{\bfn\in\mathcal{N} }H^\ddag(\bfn)$.
\end{definition}

Below, we show that, when one considers egalitarian social welfare, the minimax-regret problem indeed leads to an equal sample allocation across all groups if the signal-to-noise ratio in each group is the same. 

\begin{theorem}
Under the conditions of Theorem \ref{theorem:minimax}, 
the sample selection $\bfn^{\ddag}=\p{n^\ddag_1,n^{\ddag}_2,\cdots,n^{\ddag}_G}^\top$ with
\begin{equation}
{n}^{\ddag}_g=2\left\lfloor\frac{s_{0,g}^2+s_{1,g}^2}{2\sum_{g=1}^G\p{s_{0,g}^2+s_{1,g}^2}}N\right\rfloor, \qquad\qquad g=1,\dots,G
\label{eq:egalitarian_selection}
\end{equation}
is nearly egalitarian-minimax-regret 
in the sense that
\begin{equation}
\frac{H^\ddag\p{\bfn^\ddag}-\min_{\bfn\in\mathcal{N} }H^\ddag\p{\bfn}}{\min_{\bfn\in\mathcal{N} }H^\ddag\p{\bfn}} = \oo\p{N^{-1}}.
\label{eq:minimax_egalitarian_gap}
\end{equation}
\label{theorem:minimax_egalitarian}
\end{theorem}

One may notice that the resulting sample selection shares a common ground with the Neyman allocation (also known as the optimal allocation), which advocates sampling proportionally to the standard deviation \citep{neyman1934two}. However, Neyman's goal was to minimize the variance of a sample mean estimator, whereas our focus is on the quality of the decision and we target directly the regret with egalitarian social welfare. 

\section{Case Study: A COVID-19 Vaccine Trial}
\label{sec:case_study}

To further explore the implications of different group sample size allocation practices, we conduct a semi-synthetic analysis using historical data from COVID-19 vaccine development. 
Different demographic subpopulations may experience different benefits and side effects from vaccines, and so it is of interest to consider experimental design to inform vaccination recommendations for subpopulations. In real-world clinical trials, there are many important considerations beyond the ones discussed in this paper; our goal here is simply to help illustrate how different decision-theoretic criteria may lead to notably different group allocations. 

We use data from a large Phase 3 randomized, placebo-controlled COVID-19 vaccine trial conducted at 99 centers across the United States \citep{baden2021efficacy}. 
To account for anticipated heterogeneity in the treatment effect 
across different demographic groups, the randomization is stratified into three groups: 1. $\ge18$ and $<65$ years and not at risk (for severe Covid-19); 2. $\ge18$ and $<65$ years and at risk; 3. $\ge 65$ years. 
In an effort to 
demonstrate how the proposed sample selection schemes can be used in practice, we explore the sample selections chosen with the three minimax sample selections given in Theorems \ref{theorem:minimax}, \ref{theorem:minimax_proportional} and \ref{theorem:minimax_egalitarian}, and compare their performances under both the worst-case incidence rates and the incidence rates reported in \citet{baden2021efficacy}.

Motivated by the two-fold primary objective of the trial, we study a composite outcome $\Tilde{Y}_{g,i}=Y_{g,i,1}+\beta Y_{g,i,2}$, where $Y_{g,i,1}$ represents the incidence of severe COVID-19 of unit $i$ in group $g$, $Y_{g,i,2}$ represents the incidence of severe solicited adverse reaction of unit $i$ in group $g$ during the injections, and $\beta$ governs the tradeoff between the efficacy and the safety of the vaccine.\footnote{We define severe solicited adverse reactions as solicited adverse reactions with a toxicity grade for Erythema greater than or equal to 3 \citep{baden2021efficacy}.}
Since the incidence rates are only reported stratified by whether age is above or below 65 years, we will consider the case where there is only treatment effect heterogeneity between people with age $\ge 18$ to $<65$ yr (group 1), and with age $\ge 65$ yr (group 2). 
According to the 2020 U.S. Census Bureau's data on age \citep{census2020age}, the estimated proportion of people aged 65 and over in the United States is nearly 17\%, which motivates our choice of weights  $\alpha_1=0.83$ and $\alpha_2=0.17$. 
Throughout, we consider two different values of $\beta$, with case 1 being $\beta=0.005$ (in which case it is beneficial to assign treatment to both groups under the reported incidence rates) and case 2 being $\beta=0.025$ (in which case it is only beneficial to assign treatment to people $>=65$ years old under the reported incidence rates).

\subsection{Variance Approximation}
\label{subsec:noise_approx}

To determine the minimax-regret sample selections, one needs to specify additionally the noise levels $s_{w,g}$, $g=1,2$, $w=0,1$. 
We obtain a conservative approximation of those noise levels based on empirical observations of baseline risks for severe COVID-19 reported by the Centers for Disease Control and Prevention (CDC) \citep{cdc2020covid} and risk levels of severe adverse reactions reported in an earlier phase trial \citep{jackson2020mrna}. 

\begin{table}[t]
\centering
\begin{tabular}{|c|c|c|}
  \hline
   &&\\[-1em]
 Group & 1 & 2  \\ 
  \hline
   &&\\[-1em]
 Case 1  & 11.79 & 22.08 \\ 
   \hline
   &&\\[-1em]
 Case 2  & 11.82 & 22.10 \\ 
   \hline
\end{tabular}
\caption{Approximated noise levels $\sqrt{s_{0,g}^2+s_{1,g}^2}$ (\%) used to determine the allocations.}
\label{tab:covid_noise_approx}
\end{table}

According to CDC's data recorded as of September 2020 (during  the  period when the protocol was being prepared), the incidence rate for newly admitted patients with confirmed COVID-19 was around $0.7\%$ for people $>=18$ to $<65$ yr, and $2.5\%$ for people $>=65$ yr \citep{cdc2020covid}. Furthermore, according to a phase 1 trial of an mRNA vaccine, the incidence rate of severe solicited adverse reactions was around $6.7\%$ in a dose group of 100mcg \citep{jackson2020mrna}. Given this, we consider a conservative approximation of the noise levels, 
where we take the incidence rates of severe COVID-19 in the treated group to be the same as the ones in the control group and the incidence rates of severe solicited adverse reactions in the control group to be the same as the ones in the treated group.
Those approximated noise levels are summarized in Table \ref{tab:covid_noise_approx}.

\subsection{Budget Constraint}
\label{subsec:sample_size}

While there is no explicit requirement that sample sizes must be based on power calculations, it is generally considered good practice to conduct power calculations to determine the appropriate sample size for clinical trials \citep{food1988guideline}. 
Thus, we determine the total sample size constraint by following \citet{baden2021efficacy} that sets $N$ to be the sample size required to achieve a power of 0.9 under a specified one-sided hypothesis test with $\alpha=0.05$. 

In \citet{baden2021efficacy}, the sample size is chosen to be the minimum number of samples such that there is 90\% power to detect a 60\% reduction in hazard rate of COVID-19. 
Given that hazard ratios are commonly interpreted as the incidence rate ratios in practice \citep{hernan2010hazards}, we translate this into a sample size calculation where $N$ is chosen to be the minimum number of samples required to achieve 90\% power in detecting that $\tau_{\text{COVID}} = -0.6\,\EE{Y_{g,i,0}(0)}$. Utilizing the approximated incidence rates of severe COVID-19 in the control group as detailed in Section \ref{subsec:noise_approx}, we have the following hypothesis test:
\begin{equation}
H_0: \tau_{\text{COVID}}=0,\qquad\qquad\qquad\qquad H_1: \tau_{\text{COVID}}=-0.60\%.
\end{equation}
As in \citet{baden2021efficacy}, the sample size is calculated by assuming a homogeneous treatment effect across the entire population, which is a standard assumption in such calculations.

To calculate the total sample size, we adopt the approach outlined in \citet{charan2013calculate} and assume that $\Tilde{Y}_{g,i,1}(1)$ and $\Tilde{Y}_{g,i,1}(0)$ are generated i.i.d. from $N(b_{\text{COVID}}-\tau_{\text{COVID}}/2, s_{0,\text{COVID}}^2)$ and $N(b_{\text{COVID}}+\tau_{\text{COVID}}/2, s_{1,\text{COVID}}^2)$, respectively, with $s_{w,\text{COVID}}^2=\sum_{g=1,2}\alpha_g s_{w,g,\text{COVID}}^2$, $s_{w,g,\text{COVID}}$ denoting the standard deviations of $Y_{g,i,1}(w)$, and $b_{\text{COVID}}$ being a nuisance baseline parameter. 
We find the required total sample size $N$ that, with a difference-in-means estimator,
\begin{equation}
    \PP{\frac{\htau_{\text{COVID}}^{\text{DM}}}{\sqrt{\Var{\htau_{\text{COVID}}^{\text{DM}}}}}> Z_{95\%}\cond H_1} = 90\%.
\end{equation}
Note that under $H_1$,
\begin{equation}
\htau^{\text{DM}}_{\text{COVID}}\sim N (0.60\%, 2(s_{0,\text{COVID}}^2+s_{1,\text{COVID}}^2)/N)
\end{equation}
and $Z_{p}$ is the $p$th quantile of a standard Gaussian random variable. This gives rise to a total sample size of
\begin{equation}
N=\frac{2(s_{0,\text{COVID}}^2+s_{1,\text{COVID}}^2)\p{Z_{90\%}+Z_{95\%}}^2}{{\EE{\htau^{\text{DM}}_{\text{COVID}}\cond H_1}}^2}\approx 9320.
\label{eq:sample_size}
\end{equation}
Thus, we set the constraint to be $n_1+n_2\le 9320$, $n_1, n_2>0$. 

\subsection{Worst-Case Performance}
\label{subsec:worst_perf}

We start by comparing the worst-case performances of the minimax sample selections (Equations \eqref{eq:minimax_selection}, \eqref{eq:proportional_selection} and \eqref{eq:egalitarian_selection}) calculated based on the approximated noise levels in Table \ref{tab:covid_noise_approx}. 
For demonstration purposes, we consider a data-generating distribution as stated in Theorem \ref{theorem:minimax}, with noise levels fixed
to be the ones given in Table \ref{tab:covid_noise_approx} and treatment effects chosen adversarially. To evaluate the worst-case performances, we calculate 
\begin{equation}
    \max_\tau\EE{R(\bfn,\delta^\text{DM}(D))},
\end{equation}
the maximum expected regret using the DM estimator with group-level decisions, 
\begin{equation}
    \max_\tau\EE{R(\bfn,\Tilde{\delta}^\text{DM}(D))},
\end{equation} 
the maximum expected regret using the DM estimator with a joint decision, and 
\begin{equation}
    \max_\tau\max_g\EE{R_g(\bfn,\delta^\text{DM}(D))},
\end{equation} 
the maximum expected regret using the DM estimator with group-level decisions under an egalitarian utility function, respectively.

\begin{table}[t]
\centering
\begin{tabular}{|c|c|rrr|c|rrr|}
  \hline
  \multirow{4}{*}{Scheme} & \multicolumn{4}{c|}{Case 1 $(\beta=0.005)$} & \multicolumn{4}{c|}{Case 2 $(\beta=0.025)$} \\ 
  \cline{2-9}
  & \multirow{3}{*}{Allocation} & \multicolumn{3}{c|}{Maximum Expected} & \multirow{3}{*}{Allocation} & \multicolumn{3}{c|}{Maximum Expected}  \\ 
  & & \multicolumn{3}{c|}{Regret ($\times 10^{-4}$)} & & \multicolumn{3}{c|}{Regret ($\times 10^{-4}$)} \\ 
  \cline{3-5}
  \cline{7-9}
  & & $\delta_g,U$ & $\delta,U$ & $\delta_g,U^\ddag$ & & $\delta_g,U$ & $\delta,U$ & $\delta_g,U^\ddag$ \\
  \hline
  $18$ to $65$ yr only & (9320, 0) & $\infty$ & $\infty$ & $\infty$ & (9320, 0) & $\infty$ & $\infty$ & $\infty$  \\ 
  $>=65$ yr only & (0, 9320) & $\infty$ & $\infty$ & $\infty$ & (0, 9320) & $\infty$ & $\infty$ & $\infty$ \\ 
  Minimax & (6100, 3218) & \textbf{4.60} & $\infty$ & 9.36 & (6102, 3216) & \textbf{4.61} & $\infty$ & 9.37 \\ 
  Proportional & (7734, 1584) & 4.94 & \textbf{3.51} & 13.34 & (7734, 1584) & 4.95 & \textbf{3.51} & 13.35 \\ 
  Egalitarian & (2068, 7250) & 6.23 & $\infty$ & \textbf{6.23} & (2074, 7244) & 6.24 & $\infty$ & \textbf{6.24} \\ 
  Neyman & (3244, 6074) & 5.29 & $\infty$ & 6.81 & (3248, 6070) & 5.30 & $\infty$ & 6.82 \\ 
  \hline
\end{tabular}
\caption{Sample allocations and their worst-case expected regret ($\delta_g,U$: separate decision, aggregated utility; $\delta,U$: joint decision, aggregated utility; $\delta_g,U^\ddag$: separate decision, egalitarian utility). Recall $\Tilde{Y}_{g,i}=Y_{g,i,1}+\beta Y_{g,i,2}$; $Y_{g,i,1}$ is the incidence of severe COVID-19 of unit $i$ in group $g$; and  $Y_{g,i,2}$ represents the incidence of severe adverse reactions. }
\label{tab:covid_worst}
\end{table}

Table \ref{tab:covid_worst} presented the minimax sample allocations as well as their worst-case performances under different decision paradigms.
We first highlight that different decision paradigms result in significantly distinct sample selections. The egalitarian and the Neyman approaches allocate a substantial portion of the sample to the older group (i.e., the group with a larger anticipated variance and a smaller 
weight), with the egalitarian selection being more extreme.
In contrast, both the minimax and the proportional selections prioritize the younger group that possesses a large 
weight, with the latter assigning it a higher number of samples. 
Moreover, it is evident that the minimax-regret sample selections are indeed regret minimax, and they consistently outperform the other sample allocations with respect to the specific worst-case scenarios defined by their targeting decision paradigms.
Meanwhile, extreme selections that allocate all samples to one group result in unbounded worst case regret.
Furthermore, when the policy-maker expects to make a single decision for the entire population, any sample selection other than proportional selection may result in arbitrarily bad decisions.
Nevertheless, proportional selection's performance deteriorates when group-level decisions are permitted, particularly when the policy-maker prioritizes the utility in the worst-off group.

\subsection{Performance under Reported Incidence Rates}
\label{subsec:reported_perf}

In addition to the worst-case performances, we also evaluate the performance of the minimax sample selections under the incidence rate reported in \citet{baden2021efficacy}, reproduced in Table \ref{tab:covid_rate}.
We then use Gaussian approximation to obtain the distribution of the difference-in-means treatment effect estimator,
resulting in parameters reported in Table \ref{tab:covid_param}.

\begin{table}[t]
\centering
\begin{tabular}{|c|c|c|}
  \hline
   &&\\[-1em]
 Incidence Rate of Severe COVID-19 & $>=18$ to $<65$ yr (g=1) & $>=65$ (g=2)  \\ 
  \hline
 &&\\[-1em]
 mRNA Vaccine ($w=1$) & 0 & 0  \\ 
 &&\\[-1em]
 Placebo ($w=0$) & 0.19 & 0.28 \\ 
   \hline
 &&\\[-1em]
 Incidence Rate of Severe Solicited ARs & $>=18$ to $<65$ yr (g=1) & $>=65$ (g=2)  \\ 
  \hline
 &&\\[-1em]
 mRNA Vaccine ($w=1$) & 14.55 & 9.50  \\ 
 &&\\[-1em]
 Placebo ($w=0$) & 2.53 & 2.48 \\ 
   \hline
\end{tabular}
\caption{Reported incidence rates (\%) of severe COVID-19 and severe solicited adverse reactions in \citet{baden2021efficacy}.}
\label{tab:covid_rate}
\end{table}

\begin{table}[t]
\centering
\begin{tabular}{|c|cc|cc|}
  \hline
 & \multicolumn{2}{c|}{Case 1} & \multicolumn{2}{c|}{Case 2} \\ 
  \hline
   &&\\[-1em]
 Group & 1 & 2 & 1 & 2 \\ 
  \hline
   &&\\[-1em]
 $\tau_g$ (\%)
 & -0.13 & -0.24 & 0.11 & -0.10 \\ 
 $\sqrt{s_{0,g}^2+s_{1,g}^2}$ (\%) & 4.35 & 5.28 & 4.45 & 5.34 \\ 
   \hline
\end{tabular}
\caption{Ground truth of $\tau_g$ and $\sqrt{s_{0,g}^2+s_{1,g}^2}$ used to evaluate the performance of the allocations.}
\label{tab:covid_param}
\end{table}

\begin{table}[t]
\centering
\begin{tabular}{|c|c|rrr|c|rrr|}
  \hline
  \multirow{4}{*}{Scheme} & 
  \multicolumn{4}{c|}{Case 1 $(\beta=0.005)$} & \multicolumn{4}{c|}{Case 2 $(\beta=0.025)$} \\ 
  \cline{2-9}
  & \multirow{3}{*}{Allocation} & \multicolumn{3}{c|}{Expected Regret} & \multirow{3}{*}{Allocation} & \multicolumn{3}{c|}{Expected Regret}  \\ 
  & & \multicolumn{3}{c|}{($\times 10^{-4}$)} & & \multicolumn{3}{c|}{($\times 10^{-4}$)} \\ 
  \cline{3-5}
  \cline{7-9}
  & & $\delta_g,U$ & $\delta,U$ & $\delta_g,U^\ddag$ & & $\delta_g,U$ & $\delta,U$ & $\delta_g,U^\ddag$ \\
  \hline
  $18$ to $65$ yr only & (9320, 0) & 2.30 & 0.316 & 12.20 & (9320, 0) & 1.29 & \textbf{0.33} & 5.18  \\ 
  $>=65$ yr only & (0, 9320) & 5.38 & 0.012 & 6.47 & (0, 9320) & 4.77 & 6.76 & 5.55 \\ 
  Minimax & (6100, 3218) & \textbf{0.67} & 0.005 & \textbf{0.78} & (6102, 3216) & 1.16 & 1.73 & {2.26}  \\ 
  Proportional & (7734, 1584) & 0.75 & 0.047 & 2.36 & (7734, 1584) & \textbf{1.07} & 0.68 & 3.03  \\ 
  Egalitarian & (2068, 7250) & 1.83 & {0.003} & 2.19 & (2074, 7244) & 2.16 & 6.07 & 2.34  \\ 
  Neyman & (3244, 6074) & 1.26 & \textbf{0.002} & 1.50 & (3248, 6070) & 1.70 & 5.19 & \textbf{1.75} \\ 
  \hline
\end{tabular}
\caption{Sample allocations and their expected regret ($\delta_g,U$: separate decision, aggregated utility; $\delta,U$: joint decision, aggregated utility; $\delta_g,U^\ddag$: separate decision, egalitarian utility) under the incidence rates reported in \citet{baden2021efficacy}.}
\label{tab:covid}
\end{table}

Table \ref{tab:covid} summarizes the sample allocations and their performances. 
Despite the approximated noise levels used for guiding the selections being largely misspecified, the selections we investigated still achieve decent performances. 
In particular, although the minimax selection may appear conservative, it tends to outperform the proportional selection in most cases in this non-adversarial real-data example.
In Case 1 where the minority group exhibits a higher signal-to-noise ratio, the minimax selection that allocates more equally outperforms under separate decisions; when only a joint decision is allowed, 
the egalitarian and the Neyman selections that allocate a large portion of the sample to the easier-to-learn group achieve a smaller regret.
Nevertheless, in Case 2 where the optimal decisions for the groups differ, the egalitarian selection that aims to minimize regret in the worst-off group performs considerably worse, and it appears to be less robust to model misspecification than the Neyman selection in the setting we consider.

\section{Bayes-Optimal Sample Selection}
\label{sec:bayesian_selection}

In Section \ref{sec:minimax_selection}, we identified the minimax-regret sample selection by minimizing the expected regret with $\tau$ chosen adversarially. The minimax rule can also be interpreted as a Bayes-optimal rules with a least favorable prior on $\tau$ \citep{wasserman2013all}. While the minimax rule has attractive worst-case guarantees, a decision-maker with meaningful information about $\tau$ may prefer to conduct Bayes-optimal sample selection using their chosen prior rather than with the least favorable one underlying the minimax rule.
In this section, we briefly consider Bayesian sample selection under an informative prior, and examine alignment of the minimax-regret sample selections with Bayes-optimal ones in a simple example.

Within a Bayesian framework, the parameters are regarded as random variables governed by a prior distribution. Unlike in non-Bayesian approaches where the treatment effect $\tau$ is treated as an unknown fixed value, the experimenter holds a belief captured by a prior distribution $p$ that reflects their understanding of $\tau$ based on various sources such as historical data, literature, or common beliefs.\footnote{In the Bayesian framework, we assume knowledge of all other components of the data-generating distribution. However, it is also possible to extend our framework to incorporate priors on these components.} 
Below, we give a formal definition of the Bayes-optimal sample selection.

\begin{definition}
Given a prior $p$ on the unknown parameter $\tau$, let $\delta^p(D)$ be a decision that is made according to the sign of the posterior mean of $\tau$, i.e., 
\begin{equation}
\begin{split}
&\delta^{p}(D) = I(\htau^\text{PM}\ge 0),\\
&\htau^\text{PM} = \EE{\tau\cond D}.
\end{split}
\end{equation}
The Bayesian expected regret of a sample selection rule $\bfn \in\mathcal{N}$ subject to this prior is
\begin{equation}
\label{eq:bayes_optimal}
H^p(\bfn) = \EE[\tau\sim p]{ \EE[D]{R(\bfn,\delta^p)\cond \tau}}.
\end{equation}
Such a sample selection is Bayes-optimal with respect to prior $p$ if $\bfn \in \argmin_{\bfn\in\mathcal{N} }H^p(\bfn)$.
\label{definition:bayes_optimal}
\end{definition}

An advantage of the Bayesian decision-theoretic framework is its flexibility in modeling relationships between different groups. 
For example, one can account for information sharing across groups by adjusting the prior correlation between $\tau_g$'s. 
A high correlation implies that the treatment effects across groups are likely similar, and thus insights gained from one group can inform the others.

To obtain the Bayesian expected regret, both a prior and a data-generating distribution must be specified. When the prior is conjugate to the data-generating process, $\hat{\tau}^\text{PM}$ can be derived in closed form. Otherwise, computational methods such as Markov Chain Monte Carlo \citep{gelman2013bayesian} are needed to approximate $\hat{\tau}^\text{PM}$. In the supplementary material, we provide derivations of the posterior distribution to facilitate computation when the prior $p$ and the data-generating distribution $\mathcal{D}$ form a normal-normal conjugate.


To get some insight about the effect of prior choice on sample selection, we consider the
following synthetic two-group setting.
For unit $i$ from group $g$, $g=1,2$, we consider the following data-generating distribution:
\begin{equation}
\begin{split}
&Y_{g,i}(0) \sim N(b_g-\tau_g/2, s_{0,g}^2),\qquad Y_{g,i}(1)\sim N(b_g+\tau_g/2, s_{1,g}^2)\\
&\tau=(\tau_1,\tau_2)^\top \sim N\p{\begin{pmatrix}0\\0 \end{pmatrix},\begin{pmatrix}
    \sigma_1^2 & \rho\sigma_1\sigma_2\\
    \rho\sigma_1\sigma_2 & \sigma_2^2
\end{pmatrix}} ,\qquad p(b_g)\propto 1
\end{split}
\end{equation}
Throughout, we set $s_{0,g}=s_{1,g}=1$, and $\sigma_g=0.1$ for all $g$, and investigate the change in the resulting optimal sample selection as we vary $\alpha_1$. As in Section \ref{subsec:sample_size}, we determine the sample size using \eqref{eq:sample_size} based on a power calculation for the hypothesis test $H_0:\tau=0$ versus $H_1:\tau=0.1$, resulting in a total sample size of $N=1713$.

\begin{figure}[t]
\centering
\includegraphics[width=\linewidth]{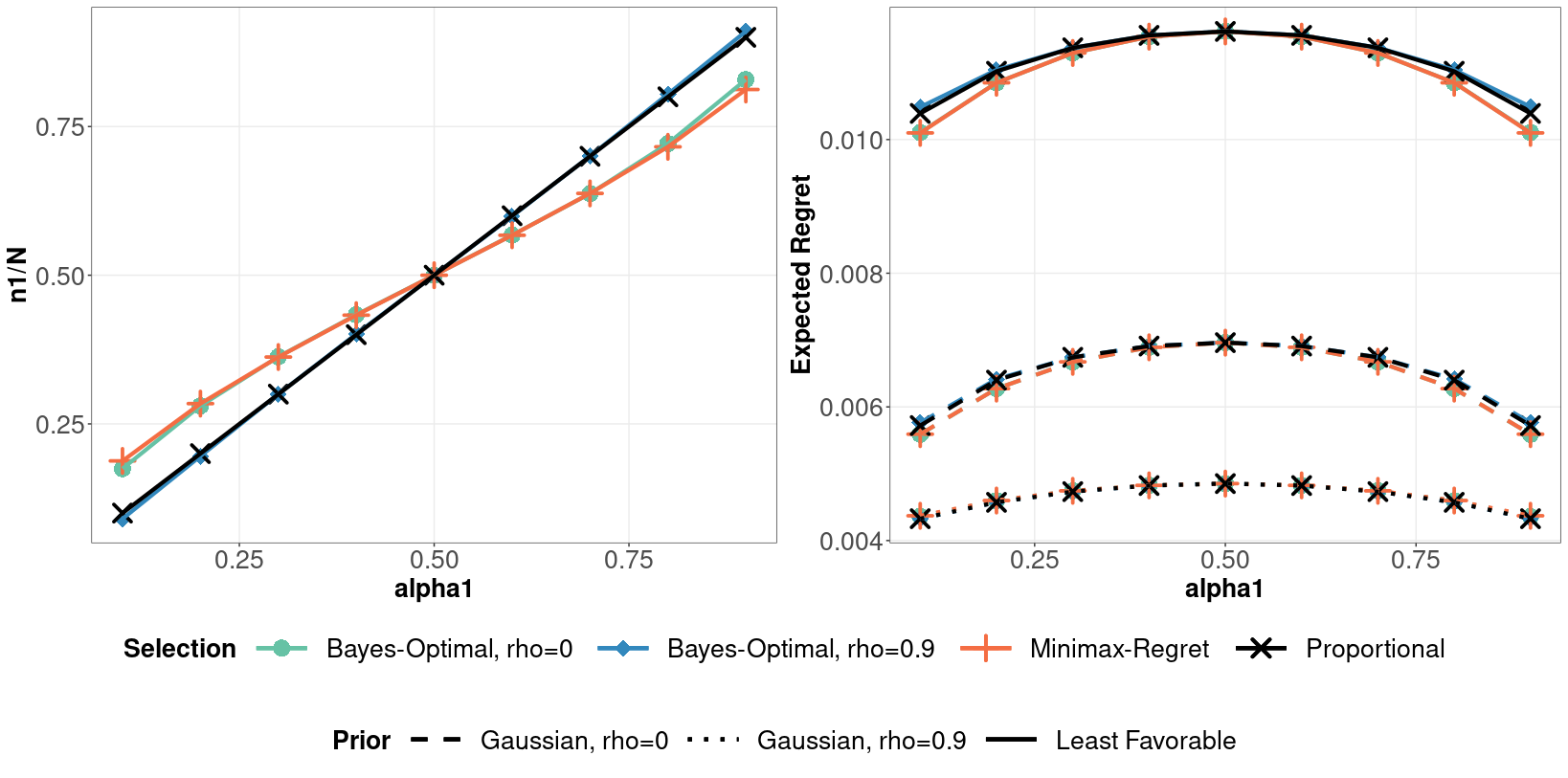}
\caption{Bayes-optimal sample selection under Gaussian prior with $\rho=0$ and $\rho=0.9$, and under the least favorable prior (minimax-regret sample selection), along with the baseline proportional selection. Left: Proportion of group 1 in sample vs. in population. Right: Bayesian expected regret of the selections under Gaussian prior with $\rho=0$ and $\rho=0.9$, and under the least favorable prior.
}
\label{fig:bayesian_synthetic}
\end{figure}

Figure \ref{fig:bayesian_synthetic} displays (1) the Bayes-optimal sample selection with $\rho=0$, (2) the Bayes-optimal sample selection with $\rho=0.9$, (3) the minimax-regret sample selection, and their corresponding performances under three different priors.
Given the specific prior and budget constraints used here, the Bayes-optimal selection under Gaussian prior with $\rho=0$ is close to the minimax-regret selection that oversamples the minority group, while the Bayes-optimal selection under Gaussian prior with $\rho=0.9$ is close to the proportional selection.\footnote{As one would expect, using $\rho\in(0,0.9)$ here leads to selections in between proportional selection and minimax selection.}
When evaluated under the least favorable prior and the Gaussian prior with $\rho=0$, the minimax-regret selection and the Bayes-optimal selection with $\rho=0$ perform better than the other two selections, especially when the minority group has a small weight. Under the Gaussian prior with $\rho=0.9$, all selections show similar performances. As anticipated, the expected regret evaluated under the prior associated with the minimax-regret selection is consistently the highest, since it is the least favorable one.

These findings suggest that, in a simple two-group setting where the sample size is chosen such that the experiment has power for typical effect sizes under the prior, our minimax rules form a reasonable starting point for a discussion on sample selection: With uncorrelated priors ($\rho = 0$) our recommendations closely match the Bayes-optimal ones, while with highly correlated priors ($\rho = 0.9$), all methods do comparably well (because it's possible to share information across groups). We do note, however, that it is also possible to construct examples where the Bayes-optimal rule diverges from our recommendations; for example, in the appendix we show a highly under-powered example in which the Bayes-optimal rule chooses to over-sample the majority group. Getting a deeper understanding of instance-level properties of the minimax-regret rules---and of when this approach is aligned with Bayesian sample selection with an informative prior---presents an interesting direction for further work.

\section{Discussion}

To guide sample selection subject to a fixed budget, we employ a minimax-regret framework that aims to minimize the maximum expected gap between the achieved utility and the maximum potential utility attainable under any decision rule had the true data-generating distribution been determined by an adversary. 
Within this minimax-regret framework, we explore the cases when the decision-maker can and cannot make separate decisions across groups, and the case when the target utility is utilitarian and egalitarian. We give a summary of all minimax-regret sample selections we investigated in Table \ref{tab:summary}. Our results highlight that optimal sample selection can vary significantly based on different practices and objectives. 

\begin{table}
    \centering
    \begin{tabular}{|c|c|c|}
        \hline
         & Aggregated Utility & Egalitarian Utility \\
        \hline
        &&\\[-1em]
        Joint Decision & $n_g\sim \alpha_g N$  &  \\
        \hline
        &&\\[-1em]
        Separate Decision & $n_g\sim\alpha_g^{2/3}N$ & $n_g\sim N/G$ \\
        \hline
    \end{tabular}
\caption{Minimax-regret sample selections under different decision paradigms.}
\label{tab:summary}
\end{table}

When considering utilitarian social welfare, the weights $\alpha_g$ stand as an important factor for determining the optimal sample selection. A prevalent approach is to sample proportionally to $\alpha_g$. Our results affirm that such proportional sampling is minimax-regret when a joint decision is guided by a sample-wise difference-in-means estimator, under which it is crucial to ensure that the sample is representative of the overall population.

One common critique of proportional sampling is that it 
fails to take into account the differentiation in subgroup variances \citep{sharma2017pros}.
Indeed, we show that when the decision-maker is not restricted to a joint decision, subgroup variances start to play an important role in optimal sample selections. Moreover, the minimax-regret solution with separate decisions suggests sampling proportionally to $\alpha_g^{-2/3}$ and thus oversamples the minority groups, highlighting the importance of making qualitative decisions for each of the subgroups in the population.

While utilitarian social welfare is a natural and common choice, it does favor groups with large
proportion weights, which may sometimes be regarded as unfair. 
An alternative is to consider the egalitarian rule that suggests making decisions that maximize the minimum utility of all individuals in society under the least undesirable condition. 
Our findings demonstrate that this indeed results in equal sample allocation across all groups when the signal-to-noise ratio is the same across all groups. In scenarios where signal-to-noise ratios vary among groups, the resulting allocation suggests sampling proportionally to the variance of the treatment-control difference, ensuring an equal level of confidence in estimating the treatment effect across all groups.

One natural concern is whether the decision-theoretic mechanisms we propose for allocating samples to subgroups---given a fixed total sample budget---may interact with or change the sample size calculations commonly done for power analyses which typically assume population-proportional group allocation in the study design. Our resulting study designs will still satisfy the power calculations, but may be with respect to a reweighted subgroup population. We leave interesting considerations of experimental design that can jointly vary the total sample size, the subgroup allocation and the randomization over treatment and control, to future work.

Finally, in this paper, we considered welfare-optimizing sample selection through the lens of minimax-regret
decision rules. This approach is attractive in terms of its robustness and generality. However, in
settings where prior evidence about the effectiveness of a treatment is available, it may be desirable
to instead use Bayesian methods that incorporate this prior knowledge.
Another interesting direction would be to investigate optimal sample selection under non-linear objectives.
Models for cooperative bargaining can be used to motivate targeting expected log-utility
\citep{nash1950bargaining,kaneko1979nash}; and \citet{baek2021fair} recently studied implications
of targeting such objectives in bandit experiments.
It may also be of interest to consider a richer welfare model that takes into account potentially different risks and benefits for
study participants (who may risk direct harms from trying out a potentially dangerous treatment) and the
general population (whose risks and benefits are mediated by the overall study findings).


\ifec

\begin{acks}
This research was supported by the Stanford Graduate School of Business initiative on Business, Government \& Society.
\end{acks}

\bibliographystyle{ACM-Reference-Format}
\bibliography{references}

\newpage

\appendix

\else

\section*{Acknowledgment}

This research was supported by the Stanford Graduate School of Business initiative on Business, Government \& Society.

\bibliographystyle{plainnat}
\bibliography{references}

\newpage

\appendix
\begin{center}
\textbf{\Large Supplemental Materials} \\ 
\end{center}
\setcounter{equation}{0}
\setcounter{figure}{0}
\setcounter{table}{0}
\setcounter{page}{1}
\makeatletter
\renewcommand{\theequation}{S\arabic{equation}}
\renewcommand{\thefigure}{S\arabic{figure}}
\renewcommand{\thetable}{S\arabic{table}}
\renewcommand{\bibnumfmt}[1]{[S#1]}
\renewcommand{\citenumfont}[1]{S#1}

\fi

\section{A two-group normal-normal example}

In this section, we derive the posterior and decision distributions for a normal-normal conjugate prior with $G=2$. Specifically, we consider the case where $p(\tau) \sim N(0, \Sigma)$ and treatments are assigned in a 1:1 completely stratified design with at most $N$ total units. We assume that $Y_{g,i}(0)$ and $Y_{g,i}(1)$ are generated i.i.d. from Gaussian distributions $N(b_g - \tau_g/2, s_{g}^2)$ and $N(b_g + \tau_g/2, s_{g}^2)$, respectively.

We start by deriving a closed-form expression of the distribution of $\htau^\text{PM}$.


\begin{align}
\PP{\tau\cond D}
&= \int_b \PP{\tau, b\cond Y,W} db\nonumber\\
&\propto \int_b \PP{Y\cond \tau, b, W}p(\tau) db\nonumber\\
&\propto p(\tau) \int_b \prod_g \exp\cb{-\sum_{i:W_{g,i}=0}\frac{\p{Y_{g,i}-b_g+\tau_g/2}^2}{2s_{g}^2} }\exp\cb{-\sum_{i:W_{g,i}=1}\frac{\p{Y_{g,i}-b_g-\tau_g/2}^2}{2s_{g}^2} } db\nonumber\\
&\propto p(\tau) \int_b \prod_g\exp\cb{-\sum_{i:W_{g,i}=0}\frac{b_g^2-2b_gY_{g,i}-b_g\tau_g+\tau_gY_{g,i}+\tau_g^2/4}{2s_{g}^2}}\nonumber\\
&\qquad\qquad\exp\cb{-\sum_{i:W_{g,i}=1}\frac{b_g^2-2b_gY_{g,i}+b_g\tau_g-\tau_gY_{g,i}+\tau_g^2/4}{2s_{g}^2} }\nonumber\\
&\propto p(\tau)\prod_g\exp\cb{\frac{\tau_g\p{\sum_{i:W_{g,i}=1}Y_{g,i}-\sum_{i:W_{g,i}=0}Y_{g,i}}-n_g\tau_g^2/4}{2s_g^2} } \nonumber\\
&\qquad\qquad \int_b \prod_g\exp\cb{-\frac{n_g}{2s_g^2}\cdot b_g^2+\frac{\sum_{i:W_{g,i}=0}Y_{g,i}+\sum_{i:W_{g,i}=1}Y_{g,i}}{s_g^2}\cdot b_g }db\nonumber\\
&\propto p(\tau)\prod_g\exp\cb{\frac{\tau_g\p{\sum_{i:W_{g,i}=1}Y_{g,i}-\sum_{i:W_{g,i}=0}Y_{g,i}}-n_g\tau_g^2/4}{2s_g^2} } \nonumber\\
&\propto p(\tau)\prod_g \exp\cb{-\frac{n_g}{8s_g^2}\p{\tau_g^2-2\cdot \frac{\sum_{i:W_{g,i}=1}Y_{g,i}-\sum_{i:W_{g,i}=0}Y_{g,i}}{n_g/2}\cdot \tau_g } }.\label{eq:pos_kernel}
\end{align}
Note that the production term in (\ref{eq:pos_kernel}) above is the kernel of a multivariate Gaussian distribution with mean
\begin{equation}
\mu_D=\p{\frac{\sum_{i:W_{1,i}=1}Y_{1,i}-\sum_{i:W_{1,i}=0}Y_{1,i}}{n_1/2},\frac{\sum_{i:W_{2,i}=1}Y_{2,i}-\sum_{i:W_{2,i}=0}Y_{2,i}}{n_2/2}}^\top,
\end{equation}
and variance
\begin{equation}
\Sigma_D=\diag\p{\frac{4s_1^2}{n_1}, \frac{4s_2^2}{n_2}}.
\end{equation}
Thus, by standard result from Bayesian inference with conjugate priors \citep{gelman2013bayesian}, 
\begin{equation}
\htau^\text{PM}=\p{\Sigma^{-1}+\Sigma_D^{-1}}^{-1} \Sigma_D^{-1}\mu_D.
\end{equation}

Since
\begin{equation}
\begin{split}
&\EE{\frac{2}{n_g}\sum_{i=1}^{n_g} W_{g,i}Y_{g,i}-\frac{2}{n_g}\sum_{i=1}^{n_g}(1-W_{g,i})Y_{g,i}\cond \tau_g}\\
&\qquad\qquad= \EE{Y_{g,i}\cond W_{g,i}=1,\tau_g}-\EE{Y_{g,i}\cond W_{g,i}=0,\tau_g}\\
&\qquad\qquad= \tau_g,
\end{split}
\end{equation}
and
\begin{equation}
\begin{split}
&\Var{\frac{2}{n_g}\sum_{i=1}^{n_g} W_{g,i}Y_{g,i}-\frac{2}{n_g}\sum_{i=1}^{n_g}(1-W_{g,i})Y_{g,i}\cond \tau_g}\\
&\qquad\qquad= \Var{\EE{\frac{2}{n_g}\sum_{i=1}^{n_g} W_{g,i}Y_{g,i}-\frac{2}{n_g}\sum_{i=1}^{n_g}(1-W_{g,i})Y_{g,i}\cond \tau_g,W_{1:n_g}}\cond \tau_g}+\\
&\qquad\qquad\qquad\qquad \EE{\Var{\frac{2}{n_g}\sum_{i=1}^{n_g} W_{g,i}Y_{g,i}-\frac{2}{n_g}\sum_{i=1}^{n_g}(1-W_{g,i})Y_{g,i}\cond \tau_g,W_{1:n_g}}\cond \tau_g}\\
&\qquad\qquad= \Var{\frac{2}{n_g}\sum_{i=1}^{n_g} W_{g,i}(b_g+\tau_g)-\frac{2}{n_g}\sum_{i=1}^{n_g}(1-W_{g,i})b_g\cond \tau_g}+\\
&\qquad\qquad\qquad\qquad \EE{\frac{4}{n_g^2}\sum_{i=1}^{n_g} W_{g,i}s_{g}^2+\frac{4}{n_g^2}\sum_{i=1}^{n_g} W_{g,i}s_{g}^2\cond \tau_g}\\
&\qquad\qquad= \Var{\frac{2}{n_g}\sum_{i=1}^{n_g} (2W_{g,i}-1)b_g+\frac{2}{n_g}\sum_{i=1}^{n_g}W_{g,i}\tau_g\cond \tau_g}+\\
&\qquad\qquad\qquad\qquad \EE{\frac{4}{n_g^2}\sum_{i=1}^{n_g} W_{g,i}s_{g}^2+\frac{4}{n_g^2}\sum_{i=1}^{n_g} W_{g,i}s_{g}^2\cond \tau_g}\\
&\qquad\qquad= \frac{4s_g^2}{n_g},
\end{split}
\end{equation}
where the last equality holds due to that $\sum_{i=1}^{n_g} W_{g,i}=n_g/2$ in a stratified completely randomized experiment with units assigned treatment or control in a 1:1 ratio,
\begin{equation}
\htau^\text{PM}\cond\tau
\sim N\p{\p{\Sigma^{-1}+\Sigma_D^{-1}}^{-1} \Sigma_D^{-1}\tau, \p{\Sigma^{-1}+\Sigma_D^{-1}}^{-1} \Sigma_D^{-1}\p{\Sigma^{-1}+\Sigma_D^{-1}}^{-1} }.
\label{eq_supp:PM_distn}
\end{equation}
Thus, it is possible to calculate directly
$\EE{I(\htau^\text{PM}_g(D;n,\tau)>0)\cond \tau_g}
=\PP{\htau^\text{PM}_g(D;n,\tau)>0\cond \tau_g }$ using the cdf of a standard Gaussian distribution.

To calculate the expected regret, we note that
\begin{equation}
\begin{split}
H^p(\bfn) 
&= \EE[\tau\sim p]{ \EE[D]{R(\bfn,\delta^p)\cond \tau}}\\
&= \sum_g\alpha_g\EE[\tau\sim p]{\tau_gI(\tau_g>0)-\tau_g\PP[D]{\htau^\text{PM}_g>0\cond \tau_g }}.
\label{eq:e_utility}
\end{split}
\end{equation}
From \eqref{eq_supp:PM_distn}, it is possible to obtain
$\PP[D]{\htau^\text{PM}_g>0\cond \tau_g }$ using the cdf of a standard Gaussian distributional function.

\section{Additional Simulation Example for Bayes-Optimal Sample Selection}
\label{sec_supp:bayes_sim}

In this section, we provide additional experiments and extend the simulation in Section \ref{sec:bayesian_selection} to a highly under-powered setting. We use the same data-generating distribution detailed in Section \ref{sec:bayesian_selection}, but with a much smaller total sample size of $N=100$.
Figure \ref{fig:bayesian_synthetic_underpowered} displays (1) the Bayes-optimal sample selection with $\rho=0$, (2) the Bayes-optimal sample selection with $\rho=0.9$, (3) the minimax-regret sample selection, and their corresponding performances under three different priors.

In contrast to the minimax-regret selection, the Bayes-optimal selections here oversamples the majority group. This is because---given the specific prior and the small budget---attempting to learn from both groups may lead to suboptimal decisions for the entire population, whereas allocating more samples to the majority group ensures that at least a significant portion of the population is treated appropriately.
Moreover, we observe increased oversampling of the majority group with greater data sharing, in which case allocating sufficient samples to one group can benefit decision-making for the entire population. 
We emphasize, however, that if the chosen prior is misspecified, the sample selections induced by the prior could result in large or even extreme regret in an adversarial environment. For example, in the right panel of Figure \ref{fig:bayesian_synthetic_underpowered}, we see that the expected regret of the $\rho=0.9$ Bayes-optimal rule explodes when evaluated against the least-favorable prior that underlies our minimax-regret rule.

\begin{figure}[t]
\centering
\includegraphics[width=\linewidth]{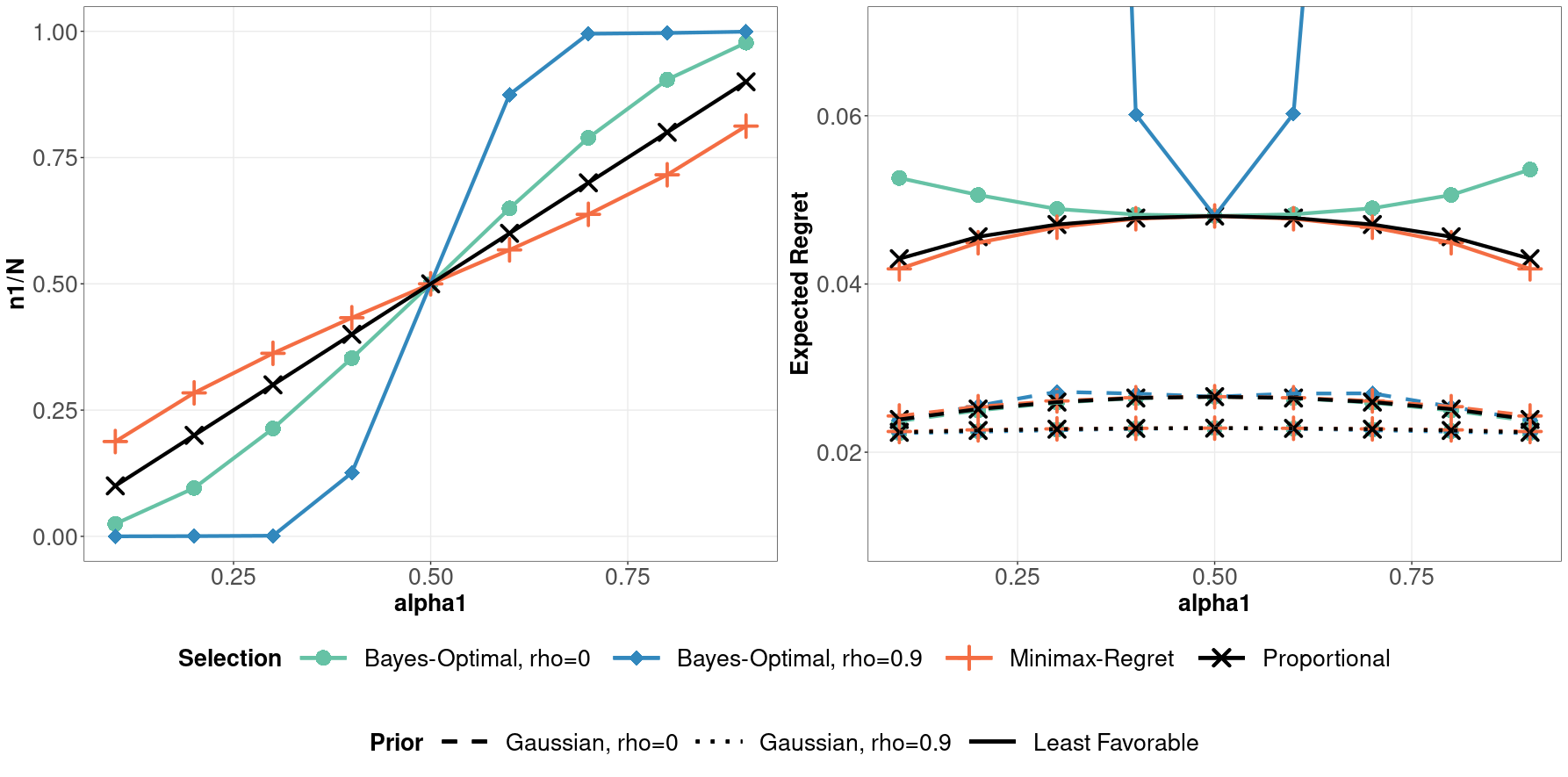}
\caption{Bayes-optimal sample selection under Gaussian prior with $\rho=0$ and $\rho=0.9$, and under the least favorable prior (minimax-regret sample selection), along with the baseline proportional selection. Left: Proportion of group 1 in sample vs. in population. Right: Bayesian expected regret of the selections under Gaussian prior with $\rho=0$ and $\rho=0.9$, and under the least favorable prior.}
\label{fig:bayesian_synthetic_underpowered}
\end{figure}

\section{Proof of Theorems}
\label{sec_supp:proof}

\subsection{Proof of Theorem \ref{theorem:minimax_proportional}}
\label{proof:minimax_proportional}

We start by finding the optimal decision rule under this restricted class of decisions. Note that 
\begin{equation}
\begin{split}
\Tilde{\delta}^* &= \argmax_{\Tilde{\delta}\in\cb{0,1}}U(\tau,\Tilde{\delta})\\
&= \argmax_{\Tilde{\delta}\in\cb{0,1}}\Tilde{\delta}\sum_{g=1}^G \alpha_g\tau_g\\
&= I\p{\sum_{g=1}^G \alpha_g\tau_g>0)}.
\end{split}
\end{equation}
Denote for shorthand $\halpha_g=n_g/N$, $\forall g$. 
Since 
\begin{equation}
\begin{split}
\EE{\htau^\text{DM,p}}
&=\sum_{g=1}^G\frac{n_g}{N}\p{\frac{1}{n_g/2}\EE{\sum_{i=1}^{n_g}W_{g,i} Y_{g,i}} - \frac{1}{n_g/2}\EE{\sum_{i=1}^{n_g}(1-W_{g,i}) Y_{g,i}}}\\
&=\sum_{g=1}^G\halpha_g\tau_g,
\end{split}
\end{equation}
\begin{equation}
\begin{split}
\Var{\htau^\text{DM,p}}
&=\sum_{g=1}^G\frac{n_g^2}{N^2}\p{\Var{\frac{1}{n_g/2}\sum_{i=1}^{n_g}W_{g,i} Y_{g,i} - \frac{1}{n_g/2}\sum_{i=1}^{n_g}(1-W_{g,i}) Y_{g,i}}}\\
&=\sum_{g=1}^G\frac{n_g^2}{N^2} \frac{2(s_{0,g}^2+s_{1,g}^2)}{n_g}\\
&=\frac{2}{N}\sum_{g=1}^G \halpha_g(s_{0,g}^2+s_{1,g}^2),
\end{split}
\end{equation}
$\htau^\text{DM,p}\sim N(\sum_g\halpha_g\tau_g,  2\sum_g\halpha_g(s_{0,g}^2+s_{1,g}^2)/N)$, and
\begin{equation}
\begin{split}
\EE{R(\bfn,\Tilde{\delta}^{DM}(D)) }
&= \EE{\sum_{g=1}^G \alpha_g\cdot\tau_g\cb{I\p{\sum_{g=1}^G \alpha_g\tau_g>0)}-I(\htau^\text{DM,p}> 0) }}\\
&= \sum_{g=1}^G \alpha_g\cdot\tau_g\cb{I\p{\sum_{g=1}^G \alpha_g\tau_g>0)}\Phi^c(t^\dag) - I\p{\sum_{g=1}^G \alpha_g\tau_g\le0} \p{1-\Phi^c(t^\dag)}},
\end{split}
\end{equation}
where
\begin{equation}
t^\dag=\frac{\sqrt{N}\sum_{g=1}^G\halpha_g\tau_g}{\sqrt{2\sum_{g=1}^G\halpha_g(s_{0,g}^2+s_{1,g}^2) }}.
\end{equation}

First, we solve the constrained maximization problem
\begin{equation}
\begin{split}
    \max_{\tau\in\RR^G}\quad & \EE{R(\bfn,\Tilde{\delta}^{DM}(D)) } \\
    \text{s.t.}\quad &
    t^\dag=\frac{\sqrt{N}\sum_{g=1}^G\halpha_g\tau_g}{\sqrt{2\sum_{g=1}^G\halpha_g(s_{0,g}^2+s_{1,g}^2) }}.
\end{split}
\end{equation}
Since the problem is fully symmetric, we will only consider the case when $\sum_{g=1}^G \alpha_g\tau_g>0$, and the case when $\sum_{g=1}^G \alpha_g\tau_g\le0$ follows with a similar derivation. When $\sum_{g=1}^G \alpha_g\tau_g>0$
\begin{equation}
\begin{split}
&\frac{\partial \EE{R(\bfn,\Tilde{\delta}^{DM}(D)) }}{\partial \tau_g}=
\alpha_g\Phi^c(t^\dag)-\alpha_g\phi(t^\dag)\cdot \frac{\sqrt{N}\halpha_g\tau_g}{\sqrt{2\sum_{g'=1}^G\halpha^2_{g'}(s_{0,g'}^2+s_{1,g'}^2) }}.
\end{split}
\end{equation}
Thus, we need to find $\tau_g$ and $\lambda$ that satisfy
\begin{equation}
\begin{split}
\tau_g&=
\frac{\alpha_g\Phi^c(t^\dag)+\lambda}{\alpha_g\phi(t^\dag)}\cdot
\frac{\sqrt{2\sum_{g'=1}^G\halpha_{g'}^2(s_{0,g'}^2+s_{1,g'}^2) }}{\sqrt{N}\halpha_g},\\
t^\dag&=\sum_g\frac{\alpha_g\Phi^c(t^\dag)+\lambda}{\alpha_g\phi(t^\dag)}.
\end{split}
\end{equation}
Solving the above system of equations gives rise to
\begin{equation}
\tau^*_g=\frac{\sqrt{2\sum_{g'=1}^G\halpha_{g'}^2(s_{0,g'}^2+s_{1,g'}^2) }}{\halpha_g\sqrt{N}}\cb{\frac{\Phi^c(t^\dag)}{\phi(t^\dag)}+\p{\frac{t^\dag}{\alpha_g}-\frac{G\Phi^c(t^\dag)}{\alpha_g\phi(t^\dag)}}\cdot\p{\sum_{g'}\frac{1}{\alpha_{g'}}}^{-1} }.
\end{equation}
Bringing $\tau^*_g$ back into $\EE{R(\bfn,\Tilde{\delta}^{DM}(D)) }$, we get
\begin{equation}
\begin{split}
&\frac{\sqrt{N}}{\sqrt{2\sum_{g=1}^G\halpha_g(s_{0,g}^2+s_{1,g}^2) }}\EE{R(\bfn,\Tilde{\delta}^{DM}(D)) } \\
&\qquad\qquad =\sum_g\frac{\alpha_g}{\halpha_g}\cb{\frac{{\Phi^c(t^\dag)}^2}{\phi(t^\dag)}\p{1-\frac{G}{\alpha_g}\cdot\p{\sum_{g'}\frac{1}{\alpha_{g'}}}^{-1} }+ 
\frac{t^\dag\Phi^c(t^\dag)}{\alpha_g}\cdot \p{\sum_{g'}\frac{1}{\alpha_{g'}}}^{-1} 
}\\
&\qquad\qquad =\cb{\sum_g\frac{\alpha_g}{\halpha_g}-G\p{\sum_{g}\frac{1}{\alpha_{g}}}^{-1}\p{\sum_g\frac{1}{\halpha_g}}}\frac{{\Phi^c(t^\dag)}^2}{\phi(t^\dag)} +
 \p{\sum_{g}\frac{1}{\alpha_{g}}}^{-1}\p{\sum_g\frac{1}{\halpha_g}}t^\dag\Phi^c(t^\dag)
\label{eq_supp:regret_prop_expression}
\end{split}
\end{equation}

Our goal is to show that 
\begin{equation}
    \argmin_{\halpha\in[0,1]^G:\sum_g \halpha_g\le 1}\max_{\tau\in\RR^G}\EE{R(\bfn,\Tilde{\delta}^{DM}(D)) }=(\alpha_1,\dots,\alpha_G)^\top.
\label{eq_supp:prop_goal}
\end{equation}
First note that ${\Phi^c(t^\dag)}^2/\phi(t^\dag)\to\infty$ as $t^\dag\to-\infty$. Thus, when 
$\sum_g\frac{\alpha_g}{\halpha_g}-G\p{\sum_{g}\frac{1}{\alpha_{g}}}^{-1}\p{\sum_g\frac{1}{\halpha_g}}\ne 0$,
the adversary can choose $\tau$ with $\sqrt{N}\sum_{g=1}^G\halpha_g\tau_g\to -\infty$ fast enough such that $\EE{R(\bfn,\Tilde{\delta}^{DM}(D)) }\to \infty$. As we have shown in proof of Lemma \ref{lemma:dm}, $\sup_{t^\dag\in\RR} t^\dag\Phi^c(t^\dag)=C_0\approx 0.17$. Thus, when $\sum_g\frac{\alpha_g}{\halpha_g}-G\p{\sum_{g}\frac{1}{\alpha_{g}}}^{-1}\p{\sum_g\frac{1}{\halpha_g}}= 0$, RHS of \eqref{eq_supp:regret_prop_expression} becomes $\p{\sum_{g}\frac{1}{\alpha_{g}}}^{-1}\p{\sum_g\frac{1}{\halpha_g}}C_0$. 

When $\halpha_g=\alpha_g, \forall g$, 
\begin{equation}
   \EE{R(\bfn,\Tilde{\delta}^{DM}(D)) }=\frac{\sqrt{2\sum_{g=1}^G\alpha_g(s_{0,g}^2+s_{1,g}^2) }}{\sqrt{N}}C_0. 
\end{equation}
Suppose by contradiction that $\argmin_{\halpha\in[0,1]^G:\sum_g \halpha_g\le 1}\max_{\tau\in\RR^G}\EE{R(\bfn,\Tilde{\delta}^{DM}(D)) }=\halpha^*$ with some $\halpha^*\ne (\alpha_1,\dots,\alpha_G)^\top$. Then $\halpha^*$ needs to satisfy that
\begin{equation}
    \sum_{g}\frac{1}{\alpha_{g}}>\sum_g\frac{1}{\halpha_g^*},
\label{eq_supp:prop_condition}
\end{equation}
{\sloppy But if \eqref{eq_supp:prop_condition} is true, there are also $\sum_g\frac{\alpha_g}{\halpha_g^*}>G$ and $G\p{\sum_{g}\frac{1}{\alpha_{g}}}^{-1}\p{\sum_g\frac{1}{\halpha_g^*}}<G$, and thus $\max_{\tau\in\RR^G}\EE{R(\bfn,\Tilde{\delta}^{DM}(D)) }\to\infty$, which is a contradiction. Thus, equation \eqref{eq_supp:prop_goal} holds and 
\begin{equation}
\argmin_{\bfn\in\widetilde{\mathcal{N}} }\max_{\tau\in\RR^{G}} \EE{R(\bfn,\Tilde{\delta}^{DM}(D)) }
=(\alpha_1N,\dots,\alpha_GN)^\top.
\end{equation}}
Again, since $\mathcal{N}\subset\widetilde{\mathcal{N}}$ and $(\alpha_1N,\dots,\alpha_GN)^\top\in \mathcal{N}$,
\begin{equation}
{n}^{\dag} \in \argmin_{\bfn\in\mathcal{N} }H^\dag(\bfn).
\end{equation}

\subsection{Proof of Theorem \ref{theorem:minimax_egalitarian}}
\label{proof:minimax_egalitarian}

First, since a separate decision for each of the groups is allowed, and both the data generation and decision processes are independent between groups, one can follow the same steps as in Proof of Lemma \ref{lemma:dm} and show that 
\begin{equation}
\inf_{\delta_g}\max_{\tau_g\in\RR} \EE[D]{R_g(n_g,\delta_g)}
=C_0\sqrt{\frac{2(s_{0,g}^2+s_{1,g}^2)}{n_g}},
\end{equation}
Thus, 
\begin{equation}
\begin{split}
H^\ddag(\bfn)
&= \inf_{\delta(D)}\max_{\tau\in\RR^{G}}\max_g \EE[D]{R_g(\bfn,\delta(D))}\\
&= \max_g\inf_{\delta(D)}\max_{\tau\in\RR^{G}} \EE[D]{R_g(\bfn,\delta(D))}\\
&= C_0\max_{g} \sqrt{\frac{2(s_{0,g}^2+s_{1,g}^2)}{n_g}}.  
\end{split}
\end{equation}
It then follows immediately that the optimal sample selection would balance the regret in each group such that
\begin{equation}
    \sqrt{\frac{2(s_{0,1}^2+s_{1,1}^2)}{n_1}}=\sqrt{\frac{2(s_{0,2}^2+s_{1,2}^2)}{n_2}}=\cdots=\sqrt{\frac{2(s_{0,G}^2+s_{1,G}^2)}{n_G}}.
\label{eq:egalitarian_condition}
\end{equation}
Combining \eqref{eq:egalitarian_condition} with the constraint $\sum_{g=1}^Gn_g= N$ and solving the system of equations yield 
\begin{equation}
\begin{split}
&\Tilde{n}_g^\ddag = \argmin_{\bfn\in\widetilde{\mathcal{N}}}
C_0\max_{g} \sqrt{\frac{2(s_{0,g}^2+s_{1,g}^2)}{n_g}}\\
&\Tilde{n}_g^\ddag = \frac{s_{0,g}^2+s_{1,g}^2}{\sum_{g=1}^G\p{s_{0,g}^2+s_{1,g}^2}}N.
\end{split}
\end{equation}
Again, combining the facts that $\min_{\bfn\in\mathcal{N} }H^\ddag\p{\bfn}=\Theta\p{N^{-1/2}}$
and
\begin{equation}
\begin{split}
&\inf_{\delta(D)}\max_{\tau\in\RR^{G}}\max_g\EE[D]{R_g(\bfn^*,\delta(D))}-\min_{\bfn\in\mathcal{N} }\inf_{\delta(D)}\max_{\tau\in\RR^{G}}\max_g \EE[D]{R_g(\bfn,\delta(D))}\\
&\qquad\qquad\le \max_g \frac{2C_0}{\sqrt{\Tilde{n}_g\cdot n^*_g}} \cdot \sqrt{\frac{2(s_{0,g}^2+s_{1,g}^2)} {\Tilde{n}_g+n^*_g}}\\
&\qquad\qquad= \oo\p{N^{-3/2}}
\end{split}
\end{equation}
gives rise to the claimed result.

\section{Proof of Lemmas}
\label{proof:lemmas}

\begin{proof}[Proof of Lemma \ref{lemma:dm}]
Since both the data generation and decision processes are independent between groups, it suffices to show that, for $g=1,\dots,G$,
\begin{equation}
\inf_{\delta_g}\max_{\tau_g\in\RR} \EE[D]{R_g(n_g,\delta_g)}
=C_0\sqrt{\frac{2(s_{0,g}^2+s_{1,g}^2)}{n_g}},
\end{equation}
and the decision rule induced by a threshold of the difference-in-means estimator in group $g$ attains this bound,
\begin{equation}
\begin{split}
&\EE[D]{R_g\p{n_g,\delta^{DM}_g}} \leq C_0\sum_{g=1}^G \alpha_g\sqrt{\frac{2(s_{0,g}^2+s_{1,g}^2)}{n_g}},
\end{split}
\end{equation}
where 
\begin{equation}
\begin{split}
&R_g\p{n_g,\delta_g}=\tau_g\cb{\delta_g^*-\delta_g },\\
&\delta_g^* = I\p{\tau_g>0} \in \argmax_{\delta_g\in\p{0,1}} \tau_g\delta_g.
\end{split}
\end{equation}

We start by noticing that the decision rule $\delta_g$ is a function of $\{Y_g,W_g\}$, where $Y_g\cond W_g=0 \sim N(b_g-\tau_g/2,s_{0,g}^2)$ and $Y_g\cond W_g=1 \sim N(b_g+\tau_g/2,s_{1,g}^2)$.
Since $\htau_g^{DM}$ is a difference in means between two independent normally-distributed samples of equal size,
$\htau_g^{DM}$ is a sufficient statistic for estimating $\tau_g$. Thus, it suffices to consider $\delta_g$ as a function of $\htau_g^{DM}$.
Besides, note that
\begin{equation}
    \htau^\text{DM}_g \sim N\p{\tau_g,2\p{s_{0,g}^2+s_{1,g}^2}/n_g }.
\end{equation}
It then follows from the results in \citet{karlin1956theory} and \citet{tetenov2012statistical} that considering the following smaller class of thresholded decision rules
\begin{equation}
    \delta_g \in \cb{I\p{\htau_g^{DM} > T}: T\in\RR}
\end{equation}
is essentially complete. This reduces our problem to showing
\begin{equation}
\inf_{T\in\RR}\max_{\tau_g\in\RR} \EE[D]{R_g\p{n_g,I\p{\htau_g^{DM} > T}}}
=C_0\sqrt{\frac{2(s_{0,g}^2+s_{1,g}^2)}{n_g}}.
\end{equation}

Next, we show that the decision rule induced by a threshold of the difference-in-means estimator
\begin{equation}
\delta^{DM}_g = I\p{\htau_g^{DM} > 0}
\end{equation} 
is regret-minimax. 
Notice that the expected regret of $\delta_g \in \cb{I\p{\htau_g^{DM} > T}: T\in\RR}$ is
\begin{equation}
\begin{split}
\EE[D]{R_g\p{n_g,\delta_g}}
&= I(\tau_g> 0)\cdot\abs{\tau_g}\cdot\PP{\htau^{DM}_g\le T}
+I(\tau_g\le 0)\cdot\abs{\tau_g}\cdot\PP{\htau^{DM}_g> T}\\
&= \abs{\tau_g}\cdot I(\tau_g> 0)\cdot
\Phi^c\p{\frac{\sqrt{n_g}\p{\abs{\tau_g}-T}}{\sqrt{2(s_{0,g}^2+s_{1,g}^2)}}}\\
&\qquad\qquad + \abs{\tau_g}\cdot I(\tau_g\le 0)\cdot
\Phi^c\p{\frac{\sqrt{n_g}\p{\abs{\tau_g}+T}}{\sqrt{2(s_{0,g}^2+s_{1,g}^2)}}}.
\end{split}
\end{equation}
where $\Phi^c(\cdot)$ is the complementary cumulative distribution function of a standard Gaussian distribution. 

We show that choosing $T=0$ is minimax-regret for any fixed non-trivial $\abs{\tau_g}$. 
Note that when $T=0$,
\begin{equation}
\begin{split}
\EE[D]{R_g\p{n_g,\delta_g}}
&= \abs{\tau_g}\cdot 
\Phi^c\p{\frac{\sqrt{n_g}\abs{\tau_g}}{\sqrt{2(s_{0,g}^2+s_{1,g}^2)}}}.
\end{split}
\end{equation}
Suppose by contradiction that $T=0$ is not minimax-regret. If the minimax-regret choice is some $T>0$, then by choosing a positive treatment effect $\tau_g>0$, the adversary achieves
\begin{equation}
\begin{split}
\EE[D]{R_g\p{n_g,\delta_g}}
&= \abs{\tau_g}\cdot 
\Phi^c\p{\frac{\sqrt{n_g}\p{\abs{\tau_g}-T}}{\sqrt{2(s_{0,g}^2+s_{1,g}^2)}}}>\abs{\tau_g}\cdot 
\Phi^c\p{\frac{\sqrt{n_g}\abs{\tau_g}}{\sqrt{2(s_{0,g}^2+s_{1,g}^2)}}},
\end{split}
\end{equation}
so it can't be that $T>0$. 
If the minimax-regret choice is some $T<0$, then by choosing a negative treatment effect $\tau_g<0$, the adversary achieves
\begin{equation}
\begin{split}
\EE[D]{R_g\p{n_g,\delta_g}}
&= \abs{\tau_g}\cdot 
\Phi^c\p{\frac{\sqrt{n_g}\p{\abs{\tau_g}+T}}{\sqrt{2(s_{0,g}^2+s_{1,g}^2)}}}>\abs{\tau_g}\cdot 
\Phi^c\p{\frac{\sqrt{n_g}\abs{\tau_g}}{\sqrt{2(s_{0,g}^2+s_{1,g}^2)}}},
\end{split}
\end{equation}
so it can't be that $T<0$ as well, which is a contradiction. Thus, the decision rule induced by a threshold of the difference-in-means estimator is minimax-regret. 

\begin{figure}[t]
\centering
\includegraphics[width=0.48\linewidth]{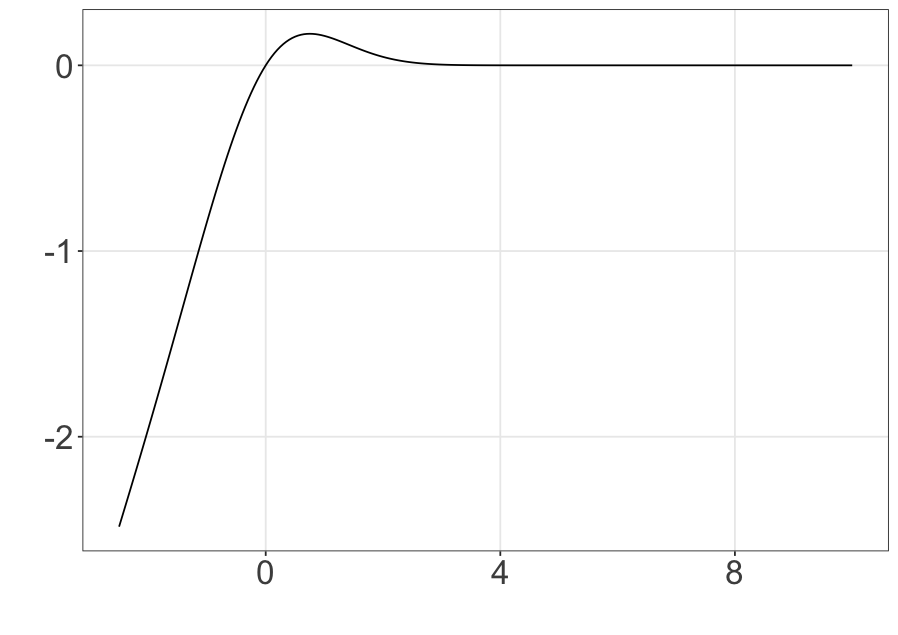}
\includegraphics[width=0.48\linewidth]{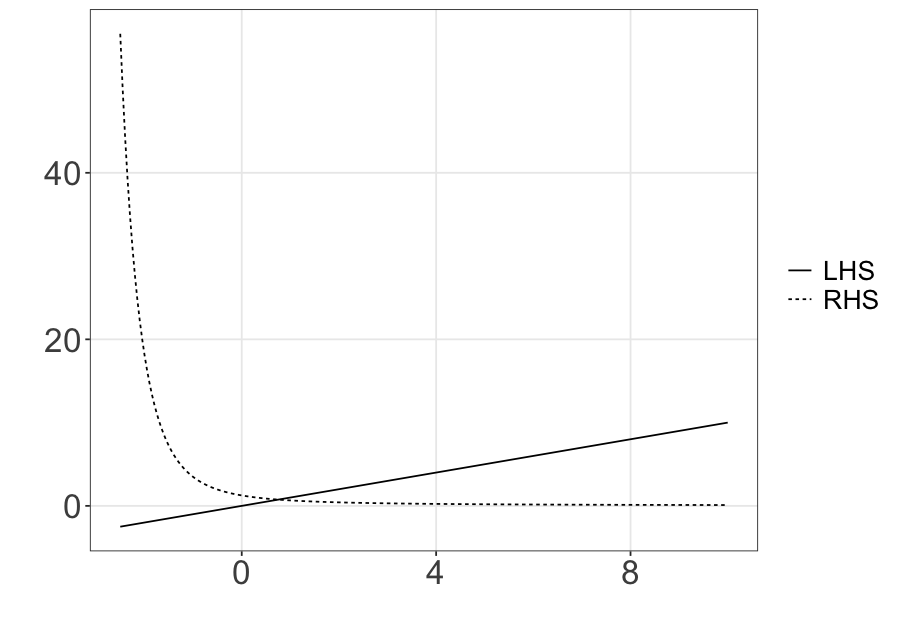}
\caption{$t_g^*$ maximizes \eqref{eq_supp:max_exp_regret_t}. Left: Function $t_g\cdot \Phi^c(t_g)$; Right: Representations of the Left Hand Side (LHS) and the Right Hand Side (RHS) of equation \eqref{eq_supp:t_eq}.}
\label{fig:minimax_proof_t_solution}
\end{figure}

Finally, we find
\begin{equation}
\max_{\tau_g\in\RR} \EE[D]{R_g(n_g,\delta_g^{DM})}.
\end{equation}
Denote for shorthand that $t_g=\frac{\sqrt{n_g}\abs{\tau_g}}{\sqrt{2(s_{0,g}^2+s_{1,g}^2)}}$. Note that $\EE[D]{R_g\p{n_g,\delta_g^{DM}}}$ can be reparametrized as
\begin{equation}
    \EE[D]{R_g\p{n_g,\delta_g^{DM}}} =  \sqrt{\frac{2(s_{0,g}^2+s_{1,g}^2)}{n_g}} \cdot t_g\cdot\Phi^c\p{t_g},
\label{eq_supp:max_exp_regret_t}
\end{equation}
and
\begin{equation}
    \max_{\tau_g\in\RR} \EE[D]{R_g\p{n_g,\delta_g^{DM}}} = \max_{t_g\in\RR_{\ge 0}} \sqrt{\frac{2(s_{0,g}^2+s_{1,g}^2)}{n_g}} \cdot t_g\cdot\Phi^c\p{t_g}.
\end{equation}
Taking derivatives with respect to $t_g$ and setting it to zero, we get
\begin{equation}
\begin{split}
t_g = \frac{\Phi^c\p{t_g}}{\phi(t_g)},
\label{eq_supp:t_eq}
\end{split}
\end{equation} 
where $\phi(\cdot)$ is the probability density function of a standard Gaussian distribution. 
Note that Equation \eqref{eq_supp:t_eq} has only a single root at $t_g^*\approx 0.75$. 
Since $t_g\cdot \phi(t_g)\to -\infty$ as $t_g\to-\infty$, $t_g\cdot \phi(t_g)\to 0$ as $t_g\to \infty$, and $t_g^*\cdot \Phi^c(t_g^*)\approx 0.17$, $t_g^*$ indeed maximizes $\EE[D]{R_g\p{n_g,\delta_g^{DM}}}$; See Figure \ref{fig:minimax_proof_t_solution} for illustrations of $t_g\cdot \Phi^c(t_g)$ and \eqref{eq_supp:t_eq}. Putting everything together, there exists a constant $C_0\approx 0.17$ such that
\begin{equation}
\inf_{\delta_g}\max_{\tau_g\in\RR} \EE[D]{R_g(n_g,\delta_g)}=C_0\sqrt{\frac{2(s_{0,g}^2+s_{1,g}^2)}{n_g}},
\end{equation}
and $\delta^{DM}_g$ attains this bound.

\end{proof}

\begin{proof}[Proof of Lemma \ref{lemma:selection}]
Since
\begin{equation}
\inf_{\delta_g}\max_{\tau_g\in\RR}\EE[D]{R_g\p{n_g,\delta_g}}= C_0\sqrt{\frac{2(s_{0,g}^2+s_{1,g}^2)}{n_g}},
\end{equation}
\begin{equation}
\inf_{\delta(D)}\max_{\tau\in\RR^G}\EE[D]{R\p{\bfn,\delta(D)}}=\sum_{g=1}^G C_0\alpha_g\sqrt{\frac{2(s_{0,g}^2+s_{1,g}^2)}{n_g}}.
\end{equation}
Recall that the constraint $\bfn\in\Tilde{\mathcal{N}}$ ensures that $\sum_{g=1}^Gn_g\le N$. Taking derivatives again with respect to $\bfn$, we get the following system of equations
\begin{equation}
\begin{split}
-\frac{C_0}{2}\alpha_g\sqrt{2(s_{0,g}^2+s_{1,g}^2)}\cdot n_g^{-3/2}
&=-\lambda,\qquad \forall g,\\
\sum_{g=1}^Gn_g&= N,
\end{split}
\label{eq_supp:lagrange}
\end{equation}
where $\lambda$ is the corresponding Lagrange multiplier for the constraint. Solving (\ref{eq_supp:lagrange}) gives rise to the claimed sample selection, with
\begin{equation}
\Tilde{n}_g=\frac{(s_{0,g}^2+s_{1,g}^2)^{1/3}\alpha_g^{\frac{2}{3}}N}{\sum_{g'=1}^G(s_{0,g'}^2+s_{1,g'}^2)^{1/3}\alpha_{g'}^{2/3}}.
\end{equation}
\end{proof}

\begin{proof}[Proof of Lemma \ref{lemma:gap_rate}]
Because $\mathcal{N}\subset\widetilde{\mathcal{N}}$, we have
\begin{equation}
C_0\sum_{g=1}^G \alpha_g\sqrt{\frac{2(s_{0,g}^2+s_{1,g}^2)}{\Tilde{n}_g}}
\leq 
\min_{\bfn\in{\mathcal{N}} }\inf_{\delta(D)}\max_{\tau\in\RR^{G}} \EE[D]{R(\bfn,\delta(D))}.
\end{equation}
Notice that $\Tilde{n}_g-2< n^*_g \le \Tilde{n}_g$. As a result,
\begin{equation}
\begin{split}
&\inf_{\delta(D)}\max_{\tau\in\RR^{G}}\EE[D]{R(\bfn^*,\delta(D))}-\min_{\bfn\in\mathcal{N} }\inf_{\delta(D)}\max_{\tau\in\RR^{G}} \EE[D]{R(\bfn,\delta(D))}\\
&\qquad\qquad\le C_0\sum_{g=1}^G \alpha_g\sqrt{\frac{2(s_{0,g}^2+s_{1,g}^2)}{n^*_g}} - C_0\sum_{g=1}^G \alpha_g\sqrt{\frac{2(s_{0,g}^2+s_{1,g}^2)}{\Tilde{n}_g}}\\
&\qquad\qquad= C_0\sum_{g=1}^G \alpha_g\sqrt{2(s_{0,g}^2+s_{1,g}^2)} \left( (n^*_g)^{-1/2} -(\Tilde{n}_g)^{-1/2} \right)\\
&\qquad\qquad= C_0\sum_{g=1}^G \alpha_g\sqrt{2(s_{0,g}^2+s_{1,g}^2)} \left( (n^*_g)^{-1/2} - (n^*_g)^{-1/2} (1 + \frac{\Tilde{n}_g - n^*_g}{n^*_g})^{-1/2} \right) \\
&\qquad\qquad\le C_0\sum_{g=1}^G \alpha_g\sqrt{2(s_{0,g}^2+s_{1,g}^2)}  (n^*_g)^{-1/2}   \left(\frac{\Tilde{n}_g - n^*_g}{2 n^*_g} \right) \\
&\qquad\qquad= \oo\p{N^{-3/2}}.
\end{split}
\end{equation}
where the last inequality follows from the generalized Bernoulli's inequality that $(1+x)^a\ge 1+ax$ for $x>-1$, $a<0$.
\end{proof}

\end{document}